\documentclass[12pt,reqno]{amsart}
\numberwithin{equation}{section}

\usepackage{amssymb}
\usepackage{enumerate, xspace}
\usepackage{marvosym} 
\usepackage[sans]{dsfont} 
\usepackage{bbm}
\usepackage{changebar}
\usepackage[colorlinks]{hyperref}

\usepackage[usenames,dvipsnames,svgnames,table]{xcolor}

\usepackage{ulem}
\usepackage{setspace}

\newtheorem{thm}{Theorem}[section]
\newtheorem{lem}[thm]{Lemma}
\newtheorem{cor}[thm]{Corollary}

\newtheorem{prop}[thm]{Proposition}

\newtheorem{rem}[thm]{Remark}

\newcommand\cF{{\mathcal F}}

\newcommand\cH{{\mathcal H}}

\newcommand\cL{{\mathcal L}}
\newcommand\cP{{\mathcal P}}
\newcommand\cO{{\mathcal O}}

\newcommand\cT{{\mathcal T}}
\newcommand\cU{{\mathcal U}}

\newcommand\bE{{\mathbb E}}
\newcommand\bN{{\mathbb N}}

\newcommand\bR{{\mathbb R}}

\newcommand\bZ{{\mathbb Z}}

\newcommand\bq{{\mathbf q}}
\newcommand\bp{{\mathbf p}}
\newcommand\be{{\mathbf e}}

\newcommand\ve{\varepsilon}

\newcommand\vf{\varphi}

\newcommand\dd{{\mathrm d}}
\newcommand\ed{{\mathrm e}}

\newcommand\RR{{\mathbb R}}
\newcommand\TT{{\mathbb T}}
\newcommand\ZZ{{\mathbb Z}}

\newcommand\R{{\mathbb R}}
\newcommand\T{{\mathbb T}}
\newcommand\Z{{\mathbb Z}}

\newcommand\Id{{\mathds{1}}}

\newcommand{\noise}{\varsigma}
\newcommand{\imaginary}{{\mathrm i}}

\newcommand{\spec}{\delta}
\newcommand{\coJ}{\mathbbm{k}}
\renewcommand\ll{\langle\!\langle}
\renewcommand\gg{\rangle\!\rangle}

\newcommand{\mc}[1]{{\mathcal #1}}

\newcommand{\mb}[1]{{\mathbf #1}}

\begin{document}

\title[Green-Kubo formula]{Green-Kubo formula for
  weakly coupled systems with noise}
\author{C.Bernardin, F.Huveneers, J.L.Lebowitz, C.Liverani and S.Olla}

\address{ C\'edric Bernardin\\
Laboratoire J.A. Dieudonn\'e
UMR CNRS 7351\\
Universit\'e de Nice Sophia-Antipolis,
Parc Valrose\\
06108 NICE Cedex 02, France}
\email{{\tt cbernard@unice.fr}}

\address{Fran\c cois Huveneers\\
 CEREMADE, UMR-CNRS 7534\\
 Universit\'{e} Paris Dauphine\\
 75775 Paris Cedex, France.}
 \email{{\tt huveneers@ceremade.dauphine.fr}}

\address{ Joel L.~Lebowitz\\
Departments of Mathematics and Physics,
Rutgers University\\
110 Frelinghuysen Road
NJ 08854 Piscataway,
USA\\
and IAS Princeton., USA.
}
\email{{\tt lebowitz@math.rutgers.edu}}

\address{ Carlangelo Liverani\\
 Dipartimento di Matematica\\
 II Universit\`{a} di Roma (Tor Vergata)\\
 Via della Ricerca Scientifica, 00133 Roma, Italy.}
 \email{{\tt liverani@mat.uniroma2.it}}

\address{Stefano Olla\\
 CEREMADE, UMR-CNRS 7534\\
 Universit\'{e} Paris Dauphine\\
 75775 Paris Cedex, France.}
 \email{{\tt olla@ceremade.dauphine.fr}}

\thanks{
We thank Herbert Spohn and David Huse for very helpful comments. 
CB, FH, JLL and SO thank Tom Spencer and Horng-Tzer Yau for the hospitality at the IAS where this work was completed, and all of us thank the BIRS-Banff. 
The research of SO and CL was founded in part by the European Advanced Grant Macroscopic Laws and Dynamical Systems (MALADY) (ERC AdG 246953). 
The research of CB was supported in part by the French Ministry of Education through the grant ANR-10-BLAN 0108 (SHEPI). 
The research of JLL was supported in part by NSF grant DMR1104500.}

\date{\today.}
\begin{abstract}
We study the Green-Kubo (GK) formula $\kappa (\ve, \noise)$
for the heat conductivity of an infinite chain of $d$-dimensional finite systems (cells) coupled by a smooth nearest neighbor potential $\ve V$. 
The uncoupled systems evolve according to Hamiltonian dynamics perturbed stochastically by an energy conserving noise of strength $\noise$. 
Noting that $\kappa (\ve, \noise)$ exists and is finite whenever $\noise > 0$, we are interested in what happens when the strength of the noise $\noise \to 0$. 
For this, we start in this work by formally expanding $\kappa (\ve, \noise)$ in a power series in $\ve$, $\kappa (\ve, \noise) = \ve^2 \sum_{n\ge 2} \ve^{n-2} \kappa_n (\noise)$
and investigating the (formal) equations satisfied by $\kappa_n (\noise)$. 
We show in particular that $\kappa_2 (\noise)$ is well defined when no pinning potential is present, 
and coincides formally with the heat conductivity obtained in the weak coupling (van Hove) limit, where time is rescaled as $\ve^{-2}t$, for the cases where the latter has been established \cite{LO, DL}.
For one-dimensional systems, we investigate $\kappa_2 (\noise)$ as $\noise \to 0$ in three cases: 
the disordered harmonic chain, the rotor chain and a chain of strongly anharmonic oscillators.
Moreover, we formally identify $\kappa_2 (\noise)$ with the conductivity obtained by having the chain between two reservoirs at temperature $T$ and $T+\delta T$, 
in the limit $\delta T\to 0$, $N \to \infty$, $\ve \to 0$.   

\end{abstract}


\maketitle

\section{Introduction}
\label{sec:introduction}

Energy transport in nonequilibrium macroscopic systems is described
phenomenologically by Fourier's law. This relates the energy flux $J$,
at the position $r$ in the system, to the temperature gradient at $r$,
via $J = -\kappa \nabla T$. The computation of the thermal
conductivity $\kappa$, which depends on the temperature and the
constitution of the system, from the underlying microscopic dynamics
is one of the central mathematical problems in nonequilibrium statistical
mechanics (see \cite{blr}\cite{LLP}\cite{dhar} and references therein).

The Green-Kubo (GK) formula gives a linear response expression for the
thermal conductivity. It is defined as the
asymptotic space-time variance for the energy currents in an infinite
system in equilibrium at temperature $T= \beta^{-1}$, evolving
according to the appropriate dynamics. It is therefore always nonnegative. For purely Hamiltonian (or
quantum) dynamics, there is no proof of convergence of the GK formula
(and consequently no proof of Fourier law).  
One way to overcome this problem is to add a dash of randomness
(noise) to the dynamics \cite{BO11}. In the present work we explore
the resulting GK formula and start an investigation of what happens
when the strength of the noise, $\noise$, goes to zero.

Our basic setup is a chain of coupled systems described in Section 2. 
Each uncoupled system (to which
we will refer as a {\it cell}) evolves according to Hamiltonian
dynamics (like a billiard, a geodesic flow on a manifold of negative
curvature, or an anharmonic oscillator...) perturbed by a dynamical
energy preserving noise, with intensity $\noise$. We will consider cases
where the only conserved quantity for the dynamics with $\noise > 0$,
is the energy. The cells are coupled
by a smooth nearest neighbor potential $\ve V$.
 We assume that the resulting infinite volume Gibbs measure has a
 convergent expansion in $\ve$ for small $\ve$. We
 are interested in
the behaviour of the resulting GK formula for $\kappa(\ve, \noise)$ given
explicitly by equation \eqref{eq:gk1} below, for small $\noise$ and
$\ve$ keeping the temperature $\beta^{-1}$ and 
other parameters fixed.

We start in Section \ref{sec:green-kubo-formula-1} by noting that for
$\noise>0$, the GK formula is well defined 
and has a finite upper bound \cite{BO11}. We do not however have a
strictly positive lower bound on $\kappa(\ve,\noise)$ except in some
special cases \cite{BO11}. We expect however that
$\kappa(\ve,\noise)>0$ whenever $\ve>0, \noise>0$, i.e. there are no
(stable) heat insulators. The situation is different when we let $\noise\to 0$. In that case we
have examples where $\kappa(\ve,\noise)\to 0$  (disordered harmonic
chains \cite{CBFH}), and where $\kappa(\ve,\noise)\to \infty$ (periodic
harmonic systems).

To make progress in elucidating the properties of
$\kappa(\ve,\noise)$, when $\noise\to 0$, we carry out in Section
\ref{sec:form-expans-green} a purely formal expansion of
$\kappa(\ve,\noise)$ in powers of $\ve$: $\kappa(\ve,\noise) =
\sum_{n\ge 2} \kappa_n(\noise) \ve^n$. This is formal for several reasons, among which is the fact that
space-time correlations entering in the GK formula involve non-local
functions and depend themselves on $\ve$.  

This, together with a similar formal expansion in Section
\ref{sec:remark-non-equil}, are the only non rigorous sections of the
paper (apart from a technical assumption at the end of Section \ref{sec:markov-operator-mc}), yet they  allow to identify the basic objects of interest.

In Section \ref{sec:markov-operator-mc} we start the study of the objects loosely introduced in Section \ref{sec:form-expans-green}. 
More precisely we show that a formal operator is in fact well defined and is a Markov generator.

We then investigate in Section \ref{sec:green-kubo-formula2} 
the structure of the term $\kappa_2(\noise)$, which we believe, but do not prove, coincides with the $\lim_{\ve\to 0}
\kappa(\ve,\noise)/\ve^2$. We show in certain cases that $\kappa_2(\noise)$ is finite
and strictly positive for $\noise>0$ by proving that it is equal to the
conductivity obtained from a weak coupling limit in which there is a
rescaling of time as $\ve^{-2}t$ (cf. \cite{LO, LOS}). We argue
further that $\lim_{\noise\to 0}\kappa_2(\noise)$ exists and is
closely related to the weak coupling macroscopic conductivity obtained for the purely 
Hamiltonian dynamics $\noise = 0$ from the beginning. The latter is computed for a geodesic flow on a surface of negative curvature, and
is strictly positive  \cite{DL}.  
The rigorous study of the zero noise limit would require the extension to random perturbations of the theory developed for deterministic perturbations in \cite{BL07, BL13}. A first step in such a direction has been achieved recently by S.Dyatlov and M.Zworski for a noise given by the full Laplace-Beltrami operator \cite{Zw}. For the noise considered here a similar result should be possible by arguing as in the discrete time case \cite{GL06}.

The identification of $\kappa_2(\noise)$ with the weak coupling limit conductivity
(suggested by H. Spohn \cite{S}, see also \cite{Huveneers}) gives some hope that the higher order
terms, can also be shown to be well defined and studied in the limit $\noise\to 0$.  
This could then lead (if nature and mathematics are kind) to a proof
of the convergence and positivity of the GK formula for a Hamiltonian
system. 

We next, formally, show in Section \ref{sec:remark-non-equil} that we obtain the
same $\kappa_2(\noise)$ for the thermal conductivity of an open
system: $N$ coupled cells in which cell $1$ and cell $N$ are in
contact with Langevin reservoirs at different temperatures, when we
let $N\to\infty$ and the two reservoir temperatures  approach to
$\beta^{-1}$. 

Section \ref{sec:ex1} is devoted to a detailed study of
$\kappa_2(\noise)$ for three examples: 
\begin{enumerate}[1)]
\item a chain of coupled pinned anharmonic oscillators;
\item a chain of rotors;
\item a harmonic chain with random (positive) pinnings.
\end{enumerate}
In  all cases we can prove that, generically, $\limsup_{\noise\to 0}
\kappa_2(\noise) < +\infty$, as 
 contrasted with the regular harmonic chain where
 $\kappa_2(\noise) \to \infty$ when $\noise \to 0$
 \cite{B08}\cite{DLK}. In case 1 we have no lower bound for this limit. 
In case 2 we expect but do not prove that $\kappa_2 (\noise)$ vanishes as $\noise$ goes to $0$. In case 3 we prove that the 
 limit of $\kappa_2(\noise)$ is zero, as it is for the conductivity $\kappa (\ve, \noise)$ of the harmonic chain with random pinning springs (\cite{CBFH}), when $\noise\to 0$.  Phase mixing, due to lack of resonances between
 frequencies of different cells at different energies, is the relevant
 ingredient for the finiteness of $\kappa_2(\noise)$ when $\noise\to 0$.

\section{The System}
\label{sec:randomdynamics}

\subsection{Cell dynamics.}

We define first the dynamics of a single uncoupled cell.
This will be given by a Hamiltonian dynamics generated by
$$
H = p^2/2 + W (q)
$$
where the position $q$ has values in some d-dimensional manifold, $q
\in M$,  and the momentum $p$ belongs to the cotangent bundle of $M$, which can be locally identified with $\bR^d$. We generally assume that $W\ge 0$, and its minimum value is $0$.
In the case of the dynamics of a billiard, $W=0$ and $M\subset \mathbb R^2$ is the corresponding
compact set of allowed positions with reflecting condition on the
boundary. Another chaotic example is given by $M$ a manifold with
negative curvature and $W=0$ (cf \cite{DL}). We will also consider
cases where $d=1$, that are completely integrable. 
Of course in the case $W=0$, the manifold $M$ will always be taken compact. If $M$ is not compact, then we ask that $W(q)$ goes to infinity with $q$ fast enough.

The Hamiltonian flow in a cell is perturbed by a noise that acts on the velocity, conserving the kinetic
energy and thus the internal energy of each cell (as in \cite{BO11}\cite{LO}).
Noises that exchange energy between different cells will not be
considered here. Consequently the energy current will be due entirely
to the deterministic interaction between the cells. 
 
The time evolved configuration $\{q(t), p(t)\}$ is given by a Markov process on the
state space  
$\Omega = M\times \bR^d$, generated by   
$$ 
 L^0 = A^0 + \noise S^0
$$
where $A^0 = p \cdot  \partial_{q} - \nabla W(q) \cdot \partial_{p}$ is the Liouville operator associated to the Hamiltonian flow and $S^0$ is the generator of the stochastic perturbation, 
which acts only on the momentum $p$ and is such that $S^0 |p|^2 = 0$. Here, $|\cdot |$ denotes the induced norm in $\RR^d$ and ``$\cdot$" the corresponding scalar product.

In dimension 1, we take an operator $S$ that generates at random
exponential times a flip on the sign of the velocity: 
\begin{equation}\label{gen-flip}
S^0 f(q,p) = \frac12[f(q, -p) - f(q,p)].
\end{equation}
More generally we can take in any dimension the random dynamics
generated by\footnote{The notation $\langle\cdot\;|\; \cF\rangle$ stands for the conditional expectation with respect to the $\sigma$-algebra $\cF$. In particular $\langle\cdot\;|\; e\rangle$ is the expectation with respect to the Liouville (microcanonical) measure on the energy surface $e=H(q,p)$.}
\begin{equation}
  \label{eq:renewalvel}
  S^0 f(q,p) = \langle f | q,e \rangle - f(q,p)
\end{equation}
where at exponential times the momentum is renewed by choosing a new
momentum with the uniform distribution $\langle \cdot | q,e\rangle$ on the sphere $\{p \in \RR^d \, ;\, |p|^2=2(e-W(q))\}$. Alternatively for $d\ge 2$, we can choose a continuous noise
by just taking for $S^0$ the Laplacian on the sphere $\{p \in \RR^d  \, ;\, |p|^2=2(e-W(q))\}$, \cite{LO}. 

In both cases $S^0$ has a spectral gap  on the
subspace of functions such that $\langle f\;|\; q, e\rangle= 0$. 
Note that the measures $\langle \cdot\;|\; q, e\rangle $ are all even in the momentum
$p$.

\subsection{Interaction}
\label{sec:interaction}

Consider now the dynamics on $\Omega^{\bZ}$ constituted
by infinitely many processes $\{ (\bq (t), \bp (t) \}:=\{q_x(t), p_x(t)\}_{x\in\bZ}$
as above, but coupled by a nearest neighbor potential $\ve V$. 
The dynamics is then generated by\footnote{ Note that, in general, we should write $V(q_x,q_{x-1})$ as $q$ might not belong to a vector space. We avoid it to simplify notation, see \cite{DL} for details.} 
\begin{equation}\label{eq:geneve}
  L_{\ve} = \sum_{x\in\bZ} \left[\noise S^x + A^x + \ve \nabla V(q_x
    - q_{x-1}) \cdot
    (\partial_{p_{x-1}} - \partial_{p_{x}})\right] = L_0 + \ve G
\end{equation}
where $L_0=A +\noise S$, $A= \sum_x A^x$, $S=\sum_x S^x$. Here $A^x$ and $S^x$ act as $A^0$ and $S^0$ on the $x$-th component of $\Omega^\ZZ$ and like the identity on the other components. For simplicity, in general we assume that the interaction potential $V$ is smooth and bounded together with its derivatives. Note however that in the special examples discussed in Section \ref{sec:ex1} we will consider also more general cases.

The energy of each cell, which is the sum of the internal energy and
of the interaction energy, is defined by 
\begin{equation}
  \label{eq:11}
  e_x^\ve = e_x + \frac \ve 2 \left(V(q_{x+1} - q_{x}) + V(q_{x} -
    q_{x-1})\right).
\end{equation}
To simplify notation we write $e_x$ for $e_x^0 = \frac{|p_x|^2}{2} + W(q_x)$, the energy of the
isolated system $x$. 
 The dynamics generated by $L_0$ preserves all the
  energies $e_x$. We denote by ${\bf e}:=\{ e_x \, ; \, x\in \ZZ\}$ the collection of the internal energies. 
 
The dynamics generated by $L_\ve$ conserves the total energy. The
corresponding energy currents $\ve j_{x,x+1}$, 
defined by the local conservation law
$$
 L_{\ve} e_x^\ve = \ve \left(j_{x-1,x}-j_{x,x+1}\right)
$$
are antisymmetric functions of the $p$'s such that
\begin{equation}
\label{current:  general formula}
 j_{x,x+1} =  -\frac 12\,  (p_x + p_{x+1}) \cdot \nabla V(q_{x+1} -q_{x}).
\end{equation}

Let us denote by $\mu_{\beta,\ve}=\langle \cdot \rangle_{\beta,\ve}$
the canonical Gibbs measure at temperature $\beta^{-1}>0$ defined by
the Dobrushin-Lanford-Ruelle equations, which of course depends on
the interaction $\ve V$. We shall assume in all the cases considered
that $\mu_{\beta,\ve}$ 
is analytical in $\ve$ for sufficiently small $\ve$ (when applied to
local bounded functions). Since we are considering here for simplicity
only the one dimensional lattice with nearest neighbor interaction, it
follows under great general conditions on $V$ and $W$ that the Gibbs
state is unique and has spatial exponential decay of correlations for
bounded local functions.

Also we assume that the \emph{equilibrium} infinite dynamics is well
defined, i.e. for a set of  initial conditions which has probability measure
one with respect to $\mu_{\beta,\ve}$ (this can be proven by standard 
techniques as in \cite{FFL, LLL, MPP, OT} and references
therein).  This defines a strongly continuous contracting semigroup of $L^2 (\mu_{\beta,\ve})$ with infinitesimal generator $L_{\ve}$ for which the smooth local functions form a core. 

We conclude this section by introducing some basic notation that will be used in the following sections.
For any given bounded local functions $f,g$, define the semi-inner product
\begin{equation}
 \label{eq:sipt}
  \ll f,g \gg_{\beta, \ve} \ =\ \sum_{x\in \ZZ} [\langle\tau_x f, g\rangle_{\beta,\ve} -
  \langle f\rangle_{\beta, \ve} \langle g\rangle_{\beta,\ve}]. 
\end{equation}
Here $\tau_x$ is the shift operator by $x$ and $\langle \cdot , \cdot \rangle_{\beta, \ve}$ the scalar product in $L^2 (\mu_{\beta,\ve})$. The sum is finite in the case $\ve =0$, and converges for $\ve>0$ thanks to the exponential decay of the spatial correlations for
local bounded functions, \cite[Proposition 8.34]{Otto}.  
Remark that, since
the velocities are always distributed independently, $\ll j_{x,x+1},j_{x,x+1}\gg_{\beta,\ve}<\infty$.

Denote by $\cH_{\ve} = L^2( \ll\cdot ,\cdot \gg_{\beta,\ve})$ the
corresponding closure. Observe that if $g$ is local, then 
$g\in \mc H_\ve$ is equivalent to $g\in \mc H_0$. In addition, if a
local function $g$ belongs to $L^2(\mu_{\beta,\ve})$ then it belongs
to $\cH_{\ve}$, even though in general $L^2(\mu_{\beta,\ve})\not\subset\cH_{\ve}$.

Note that the semigroup $e^{t L_\ve}$ is a contraction semigroup on $\cH_\ve$ as well, since, for each local functions $f$,
\[
\ll e^{t L_\ve}f, e^{t L_\ve}f\gg_{\beta,\ve}= \lim_{L\to\infty}\frac1{2L+1}\sum_{|x|, |y|\leq L}\langle e^{t L_\ve}\tau_x f,e^{t L_\ve} \tau_y f\rangle_{\beta,\ve}\leq \ll f,f\gg_{\beta,\ve}
\]
where we have used the translation invariance of the semigroup. Also one can easily check that if $f$ is a local smooth function, then
\[
\ll e^{t L_\ve}f-f, e^{t L_\ve}f-f\gg_{\beta,\ve}\leq t^2\ll L_\ve f,  L_\ve f\gg_{\beta,\ve}
\]
from which it readily follows that $L_\ve$ generates a strongly continuous semigroup on $\cH_\ve$ as well.

It is also convenient to define the following semi-inner product
\[
\ll f,g\gg_1=\ll f, (-S) g\gg_{\beta,\ve}.
\]
Let $\cH^1_\ve$ be the associated Hilbert space. We also define the Hilbert space $\cH^{-1}_\ve$ via the duality
given by the $\cH_\ve$ norm, that is  
\[
\|f\|_{-1}^2=\sup_{g} \{ \, 2 \ll f,g\gg_{\beta,\ve}-\ll g,g\gg_1 \, \} 
\]
where the supremum is taken over local bounded functions $g$.

Define the subspace of the antisymmetric functions in the velocities
\begin{equation}
 \label{eq:antisymm}
 \mc H_\ve^a = \left\{ f \in \mc H_\ve: f({\bf q},-{\bf p}) = -f({\bf q},{\bf p})\right\}.
\end{equation}
where ${\bf q} = \{q_x\in M\}_{x\in \Z}, {\bf p} = \{p_x\in \R^d \}_{x\in
 \Z}$.
Similarly define the subspace of the symmetric functions in  $\bf p$ as
$\mc H_\ve^s$.  Remark that $\cH^s_\ve\perp\cH^a_\ve$ and $\cH^s_\ve\oplus\cH^a_\ve=\cH_\ve$.
Let us define $\mc P_{\ve}^a$ and $\mc P_{\ve}^s$
the corresponding orthogonal projections, whose definition, in fact,
depends on $\ve$ only in a trivial way. Therefore we sometime omit the index $\ve$ to denote them.

Finally, given a $\sigma$-algebra $\cF$ we will use $\bE_{\beta,\ve}(\cdot\;|\;\cF)$, $\mu_{\beta,\ve}(\cdot\;|\;\cF)$ or $\langle \cdot \;|\;\cF\rangle_{\beta,\ve}$ to designate the conditional expectation. Then, for any function $f\in L^2(\mu_{\beta,\ve})$, we define
$$
(\Pi_\ve f) (\mb e) =  \mu_{\beta,\ve} (f|\mb e),\quad Q_\ve =\Id-\Pi_\ve.
$$
Note that $\Pi_\ve$ is a projector and is self-adjoint also when seen as an operator acting on $\cH_\ve$, hence it is a well defined operator both on $L^2(\mu_{\beta,\ve})$ and $\cH_\ve$.
We conclude this section with a useful Lemma.
\begin{lem}
\label{eq:gapL2}
There exists $\spec>0$ such that, for each smooth local function $f$, 
\[
\langle f , -S f\rangle_{\beta,\ve}\;\geq  \spec \langle [\, f-\langle f\;|\; {\be}, {\bq}\,  \rangle_{\beta,\ve} \, ]^2\rangle_{\beta,\ve}.
\]
\end{lem}
\begin{proof}
For any local function let $\psi_f=f-\langle f\;|\;e,q \rangle_{\beta,\ve}$, then
\[
\begin{split}
\langle \psi_f^2 \rangle_{\beta, \ve} &=\sum_{x}\left\langle \, f\, \big\{ \,  \langle f\;|\;\{e_z,q_z\}_{ z\leq x} \rangle_{\beta,\ve} -\langle f\;|\;\{e_z,q_z\}_{ z< x} \rangle_{\beta,\ve} \, \big\}\, \right\rangle_{\beta,\ve}\\
&\leq 2\sum_{x}\left\langle \,  \big[ \langle f\;|\;\{e_z,q_z\}_{ z\leq x} \rangle_{\beta,\ve} \big]^2 \,  \right\rangle_{\beta,\ve}.
\end{split}
\]
Thus, remembering the spectral gap property for each $S^x$, we get that
\[
\begin{split}
\langle f,  -S f\rangle_{\beta,\ve}\;&\geq \sum_x\langle f , (-S^x) f\rangle_{\beta,\ve} 
\geq \sum_x2\spec\, \left\langle \, \big[ \, \langle f\;|\;\{e_z,q_z\}_{ z\neq x} \rangle_{\beta,\ve} \, \big]^2 \right\rangle_{\beta,\ve}\\
&\geq \spec \langle \psi_f^2 \rangle_{\beta,\ve}^2.
\end{split}
\]
\end{proof}
If needed, the above result can be extended to a more general class of functions by density.

\section{The Green-Kubo formula}
\label{sec:green-kubo-formula-1}

The argument of Section 5 in \cite{BO11}, that will be recalled and extended shortly,
 gives the convergence of the thermal conductivity defined by the  
Green-Kubo formula\footnote{Here and in the following we work in units in which the Boltzmann constant equals one.}  
\begin{equation}
  \label{eq:gk1}
  \kappa(\ve, \noise) = \beta^2 \ve^2 \int_0^\infty \sum_{x\in \ZZ} \mathbb
  E_{\beta,\ve} \left( j_{x,x+1}(t) j_{0,1}(0) \right) dt.
\end{equation}
Here $\mathbb E_{\beta,\ve}$ indicates the expectation of the infinite
dynamics in equilibrium at temperature $\beta^{-1}$. The convergence of
 the integral in \eqref{eq:gk1} is in fact defined as 
\begin{equation}
  \label{eq:1}
  \lim_{\nu\to 0}\, \ll  j_{0,1},(\nu-L_\ve)^{-1} j_{0,1} \gg_{\beta,\ve}.
\end{equation}
Moreover, since the derivatives of $V$ are assumed to be uniformly bounded, there exists a constant $C>0$ such that
\begin{equation}
  \label{eq:15}
  \sup_{\nu>0}\, \ll j_{0,1},(\nu- L_\ve)^{-1} j_{0,1} \gg_{\beta,\ve}  
  \ \le  \ \frac C \noise .
\end{equation}
The above facts follow from the next proposition since $j_{0,1} \in {\cH}_{\ve}^a$ so that $$ \ll j_{0,1},(\nu- L_\ve)^{-1} j_{0,1} \gg_{\beta,\ve} =  \ll j_{0,1}\, , \,  {\mc P}_{\ve}^a (\nu- L_\ve)^{-1} j_{0,1} \gg_{\beta,\ve} .$$  
\begin{prop} 
\label{prop:31}
Let $\spec>0$ be the constant appearing in the statement of Lemma \ref{eq:gapL2}. There exists $\ve_0, C>0$ such that, for all $\ve\in[0,\ve_0]$, there exists a bounded operator 
$\cT_\ve$ from $\cH^{-1}_\ve$ to $\cH^1_\ve$, with norm bounded by $\noise^{-1}\spec^{-\frac 12}$,  such that
\begin{equation}
\label{eq:taepsilon}
\cT_\ve g=\lim_{\nu\to 0} \cP_{\ve}^a (\nu-L_\ve)^{-1} g.
\end{equation}
Further, for each smooth local function $h\in\cH_0^1\cap\cH_0$ such that $\cP_0^a h=0$, we have
\begin{equation}\label{eq:identity}
\cT_0 L_0h=\cT_0^*L_0^*h=0,
\end{equation}
where $L_0^*=-A+\noise S$ and $\cT_0^*=\lim_{\nu \to 0} {\cP}_0^a (\nu -L_0^*)^{-1}$ are the adjoint operators of $L_0$ and $\cT_0$ in $L^2 (\mu_{\beta,0})$.
\end{prop}
\begin{proof}
We follow a strategy put forward in \cite{BO11}.

The first step is to show that $S$ has a spectral gap on $\cH^a_\ve$. Let $f$ be a local function such that $\cP^s_\ve f=0$, note that this implies $\langle f\;|\;\bq,\be\; \rangle_{\beta,\ve}=0$. 
 Set $F_L= L^{-1/2}\sum_{|y|\leq L}\tau_y f$, and note that $\cP^s_\ve
 F_L=0$, then by Lemma \ref{eq:gapL2} we have
 \begin{equation*}
 \spec  \langle F_L^2 \rangle_{\beta,\ve} \le  \langle F_L ,  (-S)
   F_L \rangle_{\beta,\ve}
 \end{equation*}
Then since, for $L\to\infty$, we have that 
$$\langle F_L^2 \rangle_{\beta,\ve} \to \ll f, f
\gg_{\beta,\ve}, \quad \langle F_L (-S) F_L \rangle_{\beta,\ve}\to \ll
f, (-S) f \gg_{\beta,\ve},$$ the spectral gap property follows.

We then study $\nu f_\nu-L_\ve f_\nu= g$. Since $A_\ve$ inverts the
parity and $S$ preserves it, 
for each function $\vf \in {\cH}_{\ve}$ we set $\vf^+=\cP_\ve^s \vf$ and $\vf^-=\cP_\ve^a \vf$. We can then write
\begin{equation}\label{eq:oldsplit}
\begin{split}
&\nu\ll f_{\mu}^+,f_{\nu}^+\gg_{\beta,\ve}-\ll f_{\mu}^+, A f_{\nu}^-\gg+\noise\ll f_{\mu}^+,f_{\nu}^+\gg_1=\ll f_{\mu},g^+\gg_{\beta,\ve}\\
&\mu\ll f_{\nu}^-,f_{\mu}^-\gg_{\beta,\ve}-\ll f_{\nu}^-, A f_{\mu}^+\gg+\noise\ll f_{\mu}^-,f_{\nu}^-\gg_1=\ll f_{\nu},g^-\gg_{\beta,\ve}.
\end{split}
\end{equation}
Summing the above equations we have
\begin{equation}\label{eq:baseone}
\nu\ll f_{\mu}^+,f_{\nu}^+\gg+\mu\ll f_{\nu}^-,f_{\mu}^-\gg+\noise\ll f_{\mu},f_{\nu}\gg_1=\ll f_{\mu},g^+\gg_{\beta,\ve}+\ll f_{\nu},g^-\gg_{\beta,\ve}
\end{equation}
Putting $\mu=\nu$ we get
\[
\nu\ll f_{\nu}^2\gg_{\beta,\ve}+\noise\ll f_{\nu},f_{\nu}\gg_1\leq \| f_{\nu}\|_1\| g\|_{-1}.
\]
Hence $f_\nu$ is uniformly bounded in $\cH^1_\ve$ and by the spectral gap property so is  $f_\nu^-$ in $\cH_\ve$. Moreover, $\nu f_{\nu}$ converges strongly to $0$ in ${\mc H}_{\ve}$. We can then extract weakly convergent subsequences. Taking first the limit, in \eqref{eq:baseone}, $\nu\to 0$ and then $\mu\to 0$ along one such subsequences (converging to $f_*$) we have 
\[
\noise\ll f_{*},f_{*}\gg_1=\ll f_{*},g\gg_{\beta,\ve}.
\]
Taking again the limit along such subsequence, with $\mu=\nu$, we have then
\begin{equation}\label{eq:smallL2norm}
\lim_{\nu\to 0} \nu\ll f_{\nu}^2\gg_{\beta,\ve}=0.
\end{equation}
Next, taking the limit along different weakly convergent subsequences (let $f^*$ be the other limit) we have
\[
\noise\ll f_{*},f^{*}\gg_1=\ll f_{*},g^+\gg_{\beta,\ve}+\ll f^{*},g^-\gg_{\beta,\ve}
\]
and, exchanging the role of the two sequences
\[
2\noise\ll f_{*},f^{*}\gg_1=\ll f_{*},g\gg_{\beta,\ve}+\ll f^{*},g\gg_{\beta,\ve}=\noise\ll f_{*},f_{*}\gg_1+\noise\ll f^{*},f^{*}\gg_1
\]
which implies $f_*=f^*$, that is all the subsequences have the same limit. Finally, arguing similarly to the above, for $\nu(f_{\nu}-f_{\mu})+(\nu-\mu)f_{\mu}-L_\ve (f_{\nu}-f_{\mu})=0$ we obtain that the convergence takes place in the strong norm.

We are left with the proof of \eqref{eq:identity}. Note that $L_0 h= A h+\noise S h$, then $S h\in\cH_0^{-1}$ while $Ah\in\cH_0^a\subset \cH_0^{-1}$. We can thus apply the previous consideration to $g=L_0h$ and obtain, for each smooth local function $\vf\in\cH_0$,
\[
\begin{split}
\ll\vf,\cT_0 g\gg_{\beta,0}&=\lim_{\nu\to 0}\ll\cP_0^a\vf, (\nu-L_0)^{-1} L_0 h\gg_{\beta,0}\\
&=-\ll\cP_0^a\vf, h\gg_{\beta,0}+\lim_{\nu\to 0}\nu \ll(\nu-L_0^*)^{-1}\cP_0^a\vf, h\gg_{\beta,0}.
\end{split}
\]
Next, note that $\cP_0^a\vf\in\cH_0^{-1}$ and that all the above discussion applies verbatim to the operator $L_0^*$, thus \eqref{eq:smallL2norm} implies that the above limit is zero, hence the claim of the Proposition (the proof for the adjoint being the same).
\end{proof}
\begin{rem} 
\label{rem:boundedlocdf}
Since if $g \in {\mc H}^a_{\ve}$, by the spectral gap property, $\spec\|g\|_{-1}^2\leq  \ll g^2\gg_{\beta,\ve}$, it follows that $\cT_\ve$ is a bounded operator on $\cP_{\ve}^a\cH_\ve$. Also note that by essentially the same proof it is a bounded operator on $\cP_\ve^a L^2$.
\end{rem}
\begin{rem}
\label{rem:localt0}
Note that the operator ${\mc T}_0$ is a local operator in the sense that if $g$ is a local function then ${\mc T}_0g$ is also a local function.
\end{rem}

\section{Formal expansion of $\kappa(\ve, \noise)$}
\label{sec:form-expans-green}

It follows from \eqref{eq:15} that $\kappa(\ve, \noise)$  
is of order $\ve^2$, i.e.
\begin{equation}
\label{eq:limsup}
\begin{split}
&\limsup_{\ve \to 0} \ve^{-2}{\kappa} (\ve,\noise)
= \hat{\kappa}_2 (\noise)<+\infty.
\end{split}
\end{equation}
We conjecture that the limit exists and it is given by
$\kappa_2(\noise)$, the lowest term in the formal expansion of
${\kappa} (\ve,\noise)$ in powers of $\ve$:


\begin{equation}
  \label{eq:2}
  \kappa(\ve, \noise) = \sum_{n=2}^\infty
  \ve^n\kappa_n(\noise) .    
\end{equation}

It turns out that, for calculating the terms in this
expansion, it is convenient to choose $\nu = \ve^2\lambda$ in
\eqref{eq:15}, for a $\lambda >0$, and solve the resolvent equation 
\begin{equation}\label{eq:poisson}
(\lambda\ve^2-L_\ve) u_{\lambda,\ve} = \ve j_{0,1}
\end{equation}
for the unknown function $u_{\lambda,\ve}$. The reason for
considering $\lambda>0$ is to have well defined solutions also for the
infinite system. The factor $\ve^2$ is the natural scaling in view of
the subsequent computations.  In fact, in order to see an
  energy diffusion, we need to look at times of order $\ve^{-2}$, as
  suggested by the weak coupling limit \cite{LO} \cite{DL}. 

We have already remarked that $\mc P^a_\ve u_{\lambda,\ve}$ converges in
$\mc H_\ve$, for any fixed $\ve>0$.

Let us formally assume that a solution of \eqref{eq:poisson} is in the form 
\begin{equation}
  \label{eq:6}
 u_{\lambda,\ve} = \sum_{n \ge 0} (v_{\lambda,n} + w_{\lambda,n}) \ve^n,
\end{equation}
where $\Pi v_{\lambda,n} =Q w_{\lambda,n}=0$. Here $\Pi=\Pi_0$ and $Q=
Q_0$ refer to the uncoupled measure $\mu_{\beta,0}$.
This should be considered as an ansatz, and the manipulations
of the rest of this section are done without worrying about the
existence and the local regularity of the functions involved.

Given the expression (\ref{eq:6}) we can, in principle, use it in
  \eqref{eq:1} to write 
\begin{equation}\label{eq:kcoef}
\begin{split}
\beta^{-2} \kappa(\ve, \noise)& =\ve^2\lim_{\lambda\to 0}\ll
j_{0,1},(\lambda\ve^2-L_\ve)^{-1}j_{0,1}\gg_{\beta,\ve}\\ 
&=\lim_{\lambda\to 0}\sum_{n\geq0}\ve^{n+1}\ll j_{0,1},v_{\lambda,n}+w_{\lambda,n}\gg_{\beta,\ve}\\
&=\sum_{n\geq 1}\lim_{\lambda\to 0}\ve^{n}\ll j_{0,1},v_{\lambda,n-1}\gg_{\beta,\ve}
\end{split}
\end{equation}
where we have used the fact that $\ll j_{0,1},w_{\lambda,n}\gg_{\beta,\ve} =0$, by symmetry\footnote{ We will see that $w_{\lambda,0}\in \cH^s_\ve$, hence $j_{0,1}w_{\lambda,0}$ is an integrale function antisymmetric in $p$. For the higher order terms the fact that the expectation is well defined is a conjecture.} and we have, 
arbitrarily, exchanged the limit with the sum. 

Note that \eqref{eq:kcoef} is not of the type \eqref{eq:2} since the terms in the expansion depend
themselves on $\ve$. To identify the coefficients $\kappa_n$ we
would need to expand in $\ve$ also the expectations. The existence of such an expansion is not
obvious, in spite of the assumption on the Gibbs measure, since we
will see that the functions $v_{\lambda,n}$ are non local. Nevertheless, we conjecture that
\begin{equation}
  \label{eq:3}
\begin{split}
 \beta^{-2} \kappa (\ve, \noise)&={\ve }^2 \lim_{\lambda \to 0} \ll j_{0,1}, v_{\lambda,1} 
\gg_{\beta,\ve} \; +\; {o} (\ve^2)\\
&={\ve }^2 \lim_{\lambda \to 0} \ll j_{0,1}, v_{\lambda,1} 
\gg_{\beta,0} \; +\; {o} (\ve^2).
\end{split}
\end{equation}
Recalling (\ref{eq:limsup}) this yields the formula 
\begin{equation}\label{eq:kappa2}
\hat\kappa_2(\noise)=\kappa_2(\noise)= \beta^2 \lim_{\lambda \to 0}
\ll j_{0,1}, v_{\lambda,1} \gg_{\beta,0}.
\end{equation}
In fact, in Section \ref{sec:green-kubo-formula2} we will prove the second equality of \eqref{eq:3} in the special case $W=0$.

Our next task is then to find explicit formulae for $v_{\lambda,n}, w_{\lambda,n}$. Observe that $L_0 w_{\lambda,n} = 0$ for all $n\ge 0$ and that  $\Pi G\mc P^s =0$,
since $G$, defined by \eqref{eq:geneve}, changes the symmetry. 
Accordingly
\begin{equation}
  \label{eq:7}
  \begin{split}
  &v_{\lambda,0}=0\\
  & - L_0 v_{\lambda,1}- G w_{\lambda,0}=j_{0,1}\\
   &\lambda w_{\lambda,n-2}+\lambda v_{\lambda,n-2} - L_0
   v_{\lambda,n}-Gv_{\lambda,n-1}-G w_{\lambda,n-1}=0, \quad n\ge 2.
   \end{split}
\end{equation}
Let us consider first the last equation for $n=2$. Note that $\Pi L_0=0$, thus applying $\Pi$ we have, together with the second of \eqref{eq:7},
\begin{equation}\label{eq:stepone}
\begin{split}
- L_0 v_{\lambda,1}&=j_{0,1}+G w_{\lambda,0}\\
\lambda w_{\lambda,0} &=\Pi Gv_{\lambda,1}.
\end{split}
\end{equation}
Equations \eqref{eq:stepone} can be written as\footnote{The inverse of $-L_0$ should be understood as the limit of $(\nu-L_0)^{-1}$, when $\nu\to 0$, in some appropriate topology.}
\begin{equation}\label{eq:steponeb}
  \begin{split}
 \mc P^a v_{\lambda,1}&= \mc P^a(-L_0)^{-1}(j_{0,1}+ G w_{\lambda,0})\\
\lambda w_{\lambda,0} &=\Pi G \mc P^a v_{\lambda,1}.
\end{split}
\end{equation}

We can then apply $\Pi G$ to obtain\footnote{Note however that there is no obvious reason why the functions should be in the domain of $G$. Thus the objects can only be interpreted as distributions. For the moment we do not worry about this issue since we are just doing formal computations.}
\[
\begin{split}
\Pi G \mc P^a v_{\lambda,1}&= \Pi G \mc P^a
(-L_0)^{-1}\left[j_{0,1}+G w_{\lambda,0}\right]\\ 
\lambda w_{\lambda,0} &=\Pi G \mc P^a v_{\lambda,1}.
\end{split}
\]
It is then natural to consider the operator
\begin{equation}
  \label{eq:8}
  \mathcal L =   \Pi G \mc P^a (-L_0)^{-1}G\Pi.
\end{equation}
We will show in Proposition \ref{prop:Lweak} below
that the operator $\cL$ is a generator of a Markov process so that
$(\lambda- {\cL})^{-1}$ is well defined for $\lambda>0$. Hence, we have 
\begin{equation} \label{eq:w0}
\begin{split}
w_{\lambda,0}&=(\lambda-\cL)^{-1}\Pi G \mc P^a_0 (-L_0)^{-1}j_{0,1}\\
\mc P^a v_{\lambda,1}&= \mc P^a (-L_0)^{-1}\left[j_{0,1}+G w_{\lambda,0}\right].
\end{split}
\end{equation}
We can now analyse the case $n>2$. 
If we apply $\Pi$ and $Q$ at the last equation of \eqref{eq:7} we obtain
\[
\begin{split}
\lambda w_{\lambda,n-2}&=\Pi Gv_{\lambda,n-1}\\
L_0 v_{\lambda,n}&=\lambda v_{\lambda,n-2} -QGv_{\lambda,n-1}-G w_{\lambda,n-1}.
\end{split}
\]
Arguing as before we have
\begin{equation}
  \label{eq:hereitis}
\begin{split}
&w_{\lambda,n}=(\lambda-\cL)^{-1}\Pi G (-L_0)^{-1}\left[-\lambda
  v_{\lambda, n-1} + QG v_{\lambda,n}\right], \qquad n\ge 1\\
&v_{\lambda,n+1}= (-L_0)^{-1}\left[-\lambda v_{\lambda, n-1} + G
  w_{\lambda,n} + QG v_{\lambda,n}\right], \qquad n\ge 1.
\end{split}
\end{equation}
In the following we will discuss explicitly only $\cP^a v_{\lambda,1}$ and $w_{\lambda,0}$, showing that they are well defined. The study of the higher order terms is harder due to the presence of the $QGv_{\lambda,n}$ which, on the one hand, depends on the symmetric part of $v_{\lambda,n}$ (on which we have poor bounds) and, on the other hand, has no obvious reason to be in $\cH^{-1}$ (on which we know how to invert $L_0$). In order to make further progresses one must, at least, extend Proposition \ref{prop:31} to all the functions that have zero average with respect to each microcanonical measure. This, in principle, can be done in specific cases (e.g., see \cite{LO}) but at the price of dealing with less conventional functional spaces, hence making the argument much more delicate.



\section{The operator ${\cL}$  }
\label{sec:markov-operator-mc}

In this section we rigorously identify the quadratic (Dirichlet) form
associated to the operator $\cL= \Pi G  {\mc T}_0 G\Pi$ defined in \eqref{eq:8}. 

Let us denote by $\rho_{\beta} (d \be)$ the distribution of the internal
energies $\be=\{ e_x \, ; \, x \in \ZZ \}$ under the Gibbs measure
$\mu_{\beta,0}$. It can be written in the form  
\begin{equation}
d\rho_{\beta} (\be)= \prod_{x \in \ZZ} Z_\beta^{-1} \exp( - \beta e_x -U(e_x)) de_x\label{eq:rhobeta}
\end{equation}
for a suitable function $U$. We denote the formal sum $\sum_x U(e_x)$
by ${\mc U}:={\mc U} (\be)$. We denote also, for a given value of the
internal energy ${\tilde e}_x$ in the cell $x$,  by $\nu_{{\tilde e}_x}^{x}$ the microcanonical probability measure in the cell
$x$, i.e. the uniform probability measure on the manifold
$$
\Sigma_{{\tilde e}_x}:=\{ (q_x,p_x) \in \Omega \, ; \, e_x (q_x,p_x)
\, =\, {\tilde e}_x \}.
$$ 
Note that it is not obvious that the operator is well defined in $\cH_\ve$ since we do not know if $\cT_0$ maps the range of $G$ in its domain, yet, by Proposition \ref{prop:31}, it is well defined as an operator from smooth local functions to distributions.

We are going to identify $\cL$ as the Markov generator of a Ginzburg-Landau dynamics
 \begin{equation}
 {\cL}_{\rm{GL}} = \sum_x e^{\cU} (\partial_{e_{x+1}} - \partial_{e_x} )
 \left[e^{-\cU} \gamma^2 (e_x, e_{x+1}) 
(\partial_{e_{x+1}}  - \partial_{e_x} ) \right],\label{eq:12}
  \end{equation}
where
\begin{equation}
\label{gamma square}
  \gamma^2 (e_0, e_{1})= \int_{\Sigma_{e_0} \times
    \Sigma_{e_{1}}} \big(  j_{0,1}\;  {\mc T}_0\,   j_{0,1} \big) \; d\nu_{e_0}^0 d\nu_{e_{1}}^1. 
\end{equation}
We recall that ${\mc T}_0$ is defined by (\ref{eq:taepsilon}) and that ${\mc T}_0 j_{0,1}$ is a local function. Since $\nabla V$ is uniformly bounded, $j_{0,1} \in L^{2} ( \mu_{\beta,0})$ so that $j_{0,1} \in L^{2} ( \nu_{e_0}^0 \otimes \nu_{e_{1}}^1)$ for almost every $e_0,e_1$. Recalling that ${\mc T}_0$ is a bounded operator on ${\mc P}^a L^2$ (see Remark \ref{rem:boundedlocdf}), we conclude that ${\mc T}_0 j_{0,1} \in L^2  ( \mu_{\beta,0})$ and thus that ${\mc T}_0 j_{0,1} \in L^{2} ( \nu_{e_0}^0 \otimes \nu_{e_{1}}^1)$ for almost every $e_0,e_1$. Therefore, $\gamma^2 (e_0,e_1)$ is finite for almost every $e_0, e_1$ and $\gamma$ belong both to $L^2$ and $\cH_\ve$.

Formally the previous formula reads
\begin{equation}
\label{gamma square1}
  \gamma^2 (e_0, e_{1})= \int_0^\infty dt\;  \int_{\Sigma_{e_0} \times
    \Sigma_{e_{1}}}  j_{0,1}\;  (e^{t L_0} j_{0,1}) \; d\nu_{e_0}^0 d\nu_{e_{1}}^1 
\end{equation}
where $e^{t L_0}$ denotes the semigroup of the uncoupled dynamics
generated by $L_0$.  

The operator ${\cL}_{\rm{GL}} $ is well defined only if $\gamma^2$ has some regularity
properties, that are actually proven in specific examples
\cite{LO}\cite{DL}.
In the generality we consider in this work we can only prove that
\eqref{gamma square} is well defined and that the corresponding drift
\begin{equation}\label{eq:alpha}
\alpha(e_x, e_{x+1}) = e^{\mc U({\bf e})} (\partial_{e_{x+1}}
    - \partial_{e_x} ) \left[e^{-\mc U({\bf e})} \gamma(e_x,
      e_{x+1})^2\right]. 
\end{equation}
is a well defined distribution. 
So we show that the Dirichlet forms associated to $\cL$ and $\cL_{GL}$ coincide. Then in the cases where $\gamma^2$ is proven smooth \eqref{eq:12} is well defined and $\cL=\cL_{GL}$. A simple computation shows that
\begin{equation}
\label{eq:DF-GL}
  \begin{split}
    \ll g, (-{\cL}_{\rm{GL}})  f \gg_{\beta,0} &= \left\langle \gamma^2 (e_0, e_1)
      \left[(\partial_{e_1} 
        -\partial_{e_0})\Gamma_f\right] \left[(\partial_{e_1}
        -\partial_{e_0})\Gamma_g \right]\right\rangle_{\beta,0}
  \end{split}
\end{equation}
where for any function $f$ of the energies, we denote by $\Gamma_f = \sum_x
\tau_x f$ (intended as a formal sum). 

\begin{prop}\label{prop:Lweak}
For each local smooth functions $f,g$ of the energies only we have
\begin{equation}\label{dir-form}
\ll g ,(-\cL) f \gg_{\beta,0} = \ll g, (-\cL_{\rm{GL}}) f \gg_{\beta,0}. 
  \end{equation}
In addition, $\alpha (e_x,e_{x+1}) =\Pi G \mc T_0\,   j_{x,x+1}$ 
is a well defined distribution equal to \eqref{eq:alpha}.
\end{prop}

\begin{proof}

Let us start by noting that, for each local smooth function $f$ depending
only on the energies,
\begin{equation}\label{eq:G-action}
\begin{split}
Gf&=-\sum_x \nabla V(q_x-q_{x-1})\cdot
(p_x\partial_{e_x}-p_{x-1}\partial_{e_{x-1}})f\\ 
&=\sum_x j_{x-1,x}(\partial_{e_x}-\partial_{e_{x-1}})f\\
&\quad\quad-\frac 12\sum_x[L_0V(q_x-q_{x-1})](\partial_{e_{x}}+\partial_{e_{x-1}})f\\
&=\sum_x j_{x-1,x}(\partial_{e_x}-\partial_{e_{x-1}})f- L_0\omega
\\
&=\sum_x j_{x-1,x}(\partial_{e_x}-\partial_{e_{x-1}})f+ L_0^*\omega,
\end{split}
\end{equation}
where
\[
\omega=\frac 12\sum_x V(q_x-q_{x-1})(\partial_{e_{x}}+\partial_{e_{x-1}})f.
\]
Since $j_{0,1}(\partial_{e_0}-\partial_{e_{1}})f \in \mc H_0^{-1}$ and $\omega\in\cH_0\cap\cH^1$ with $\cP^a\omega=0$, Proposition \ref{prop:31} implies that $\mc T_0 \, Gf$ is
in $\mc H^a_0$ and it is equal to  
\begin{equation}
\label{eq:GL}
  \begin{split}
   {\mc T}_0 Gf \, = \,  \sum_x [\mc T_0\,  j_{x-1,x}] (\partial_{e_x}-\partial_{e_{x-1}})f. 
\end{split}
\end{equation}

Next, note that the adjoint of $G_y$ in $L^2 (\mu_{\beta,0})$ is given by
\begin{equation}\label{eq:G-adjoint}
  \begin{split}
    G_y^*&=-G_y-\beta\nabla V(q_y-q_{y-1})\cdot
    (p_y-p_{y-1})\\
    &=-G_y-\beta L_0 V(q_y - q_{y-1}) = -G_y + \beta L_0^* V(q_y - q_{y-1}).
  \end{split}
\end{equation}
and that $G_y^* g \in \mc H^a_0$. It follows that the adjoint of $G=\sum_y G_y$ in ${\mc H}_0$ (that we still denote by $G^*$) is given by 
$$-G + \beta \sum_y L_0^* V(q_y - q_{y-1}).$$

Thus, if $f,g$ are smooth local functions of the energies only then we have
\[
\begin{split}
  \ll g\, , \, \cL f \gg_{\beta,0} =&\ll G^* g\, , \, {\mc T}_0 Gf
  \gg_{\beta,0} \\
  =& -\ll G g\, , \,  {\mc T}_0 Gf
  \gg_{\beta,0} + \beta\, \ll \,  \sum_y L_0^* V(q_y- q_{y-1}) g\, , \,  {\mc T}_0 Gf
  \gg_{\beta,0}\\
 =&-\ll G g\, , \,  {\mc T}_0 Gf
  \gg_{\beta,0}
\end{split}
\]
where we have used Proposition \ref{prop:31} again.

By using (\ref{eq:GL}), (\ref{eq:G-action}) and Proposition \ref{prop:31} one last time we get then 
\[
\begin{split}
  \ll g\, , \, \cL f \gg_{\beta,0} =& -\sum_{x,y}\ll
  \, j_{y-1,y}(\partial_{e_y}-\partial_{e_{y-1}})g \, , \, 
  {\mc T}_0 j_{x-1,x}(\partial_{e_x}-\partial_{e_{x-1}})f
  \, \gg_{\beta,0}
  \end{split}
\]
where in the expressions above the sums are restricted to the finite
support of the corresponding functions. This can be rewritten as
\[
 \begin{split}
 & \ll g\, , \, \cL f \gg_{\beta,0} \\
 &= -\sum_{x,y,z}\langle
  \, j_{y+z-1,y+z} \tau_z \big[ (\partial_{e_y}-\partial_{e_{y-1}})g \big] \, , \, 
  {\mc T}_0 j_{x-1,x}(\partial_{e_x}-\partial_{e_{x-1}})f
  \, \rangle_{\beta,0}
  \end{split}
\]
Note that ${\mc T}_0 j_{x-1,x}$ is a local function depending only on $p_{x-1}, q_{x-1}, p_x,q_x$ and that we have
\[
\langle j_{y-1,y} {\mc T}_0 j_{x-1,x}|{\mb e} \rangle_{\beta,0}=
 \delta_{xy}\gamma^2 (e_{x-1}, e_x),
\]
and consequently we obtain that the Dirichlet form $\ll g \, , \, -\cL f \gg_{\beta,0}$ coincides with the RHS of (\ref{eq:DF-GL}). Equation \eqref{eq:alpha} is proved by similar arguments.
%
\end{proof}

We conclude this section by noting that 
the operator ${\mc L}_{\rm{GL}}$ is the generator of a 
Ginzburg-Landau dynamics which is reversible with respect to 
$\rho_{\beta}$, for any $\beta>0$. 
It is conservative
in the energy $\sum_x  e_x$ and the microscopic current corresponding to this conservation
law is given by the righthand side of (\ref{eq:alpha}). The
corresponding finite size 
dynamics appears in \cite{LO,DL} as the weak coupling limit of a
finite number $N$ (fixed) of cells weakly coupled by  a potential
$\epsilon V$ in the limit $\epsilon \to 0$ when time $t$ is rescaled
as $t \ve^{-2}$.

\begin{rem}\label{rem:core}
It should be possible to apply existing arguments to prove that $\cL_{GL}$, and hence $\cL$, is a closed operator also in $\cH_0$ and that the local smooth functions constitute a core of self-adjointness of $\mathcal L$ in $\cH_0$ (see \cite{F2} and \cite{S0}, where this
statement is proven for the more difficult model of interacting
brownian motions). A proof of this fact exceeds the scopes of the present article, thus in the following we will simply assume it.
\end{rem}
The hydrodynamic limit of the Ginzburg-Landau dynamics is then given (in the diffusive scale $tN^2$,
$N \to + \infty$), by a heat equation with diffusion coefficient
 which coincides with $\kappa_2$ as given by \eqref{eq:gke} below
 (\cite{Var},\cite{LOS}).

\section{The lowest order term $\kappa_2(\noise)$}
\label{sec:green-kubo-formula2}

In this section we restrict our study to the unpinned case $W=0$, obtaining an explicit formula for $\kappa_2(\noise)$.  In particular we prove, for this case, the second equality of \eqref{eq:3}. Note however that if one accepts the formula \eqref{eq:kappa2} for  $\kappa_2(\noise)$, then the following arguments, with $\ve=0$, yield an explicit formula for $\kappa_2(\noise)$ in quite some generality. 

We assume also in the following that $\gamma_{0,1}^{-1}\in L^2(\mu_{\ve,\beta})$, such an assumption is verified in all the examples discussed in this paper.
Also we assume that $\gamma^2$ is smooth, as proven in the examples considered in \cite{LO} and \cite{DL}. Thus we have ${\mc L}={\mc L}_{\rm{GL}}$.

Observe that by the definition of $\rho_\beta$ (given by \eqref{eq:rhobeta}), if $f$ depends only on the energies, we have  
$\langle f\rangle_{\beta,0} =\rho_{\beta} (f)$. In the following we
will denote by $\langle\cdot\rangle_\beta$ the integration with respect
to $\rho_{\beta}$. Moreover the semi-inner product $\ll \cdot, \cdot
\gg_{\beta}$ corresponding to $\langle \cdot \rangle_{\beta}$ is
defined as in (\ref{eq:sipt}).

To simplify notation, we denote $\gamma (e_x,e_{x+1})$
(resp. $\alpha(e_x, e_{x+1})$)  by $\gamma_{x,x+1}$
(resp. $\alpha_{x,x+1}$) and $w_{\lambda,0}$ by $w_\lambda$. Define
the operator $D_{x,x+1} = 
\gamma_{x,x+1} (\partial_{e_{x+1}} 
- \partial_{e_x})$,
 then the adjoint, with respect to $\rho_\beta$, is given by
\begin{equation*} 
  D^*_{x,x+1} =- e^{-\cU}  (\partial_{e_{x+1}} - \partial_{e_x}) e^{\cU}\gamma_{x,x+1}
\end{equation*}
and consequently $\cL = \sum_x D_{x,x+1}^* D_{x,x+1}$. First note that
for any $\lambda>0$,  we have by \eqref{eq:w0}
\begin{equation}
\lambda  w_{\lambda}-\cL w_{\lambda}= \alpha_{0,1}
= D_{0,1}^* \gamma_{0,1} .\label{eq:res}
\end{equation}
This resolvent equation, which involves only functions of the
energies, has a well defined solution $w_{\lambda} \in \mc H_0$ for
any $\lambda >0$.   

For each smooth local functions $f,g$,  relation  \eqref{eq:DF-GL} imply
\begin{equation}
  \begin{split}
    \ll f, -\cL g \gg_{\beta} &= \left\langle \gamma^2_{0,1}
      \left[(\partial_{e_1} 
        -\partial_{e_0})\Gamma_f\right] \left[(\partial_{e_1}
        -\partial_{e_0})\Gamma_g \right]\right\rangle_{\beta} \\
    &=
    \left \langle  (D_{0,1} 
      \Gamma_f) (D_{0,1}\Gamma_g) \right \rangle_{\beta}, \\
    \ll f , \alpha \gg_{\beta} &= - \left\langle \gamma_{0,1}
      D_{0,1}\Gamma_f \right \rangle_{\beta} .
  \end{split}
\end{equation}
Since $w_\lambda$ is in the domain of $\cL$, by Remark \ref{rem:core} and equation \eqref{eq:res} we have
\begin{equation}
  \lambda \ll w_{\lambda}, w_{\lambda} \gg_{\beta} + \left \langle (D_{0,1}
      \Gamma_{w_{\lambda}})^2 \right \rangle_{\beta} = - \left \langle \gamma_{0,1}
      D_{0,1}\Gamma_{w_{\lambda}} \right \rangle_{\beta} 
\end{equation}
thus by Schwarz inequality
\begin{equation*}
  \lambda \ll w_{\lambda}, w_{\lambda} \gg_{\beta} + \left \langle (D_{0,1}
      \Gamma_{w_{\lambda}})^2 \right \rangle_{\beta} \le  \left
      \langle \gamma_{0,1}^2 \right\rangle_\beta^{1/2} \left\langle  
      (D_{0,1}\Gamma_{w_{\lambda}})^2 \right \rangle_{\beta}^{1/2}
\end{equation*}
and this gives the bounds
\begin{equation}\label{eq:Gamma-bound}
  \lambda \ll w_{\lambda}, w_{\lambda} \gg_{\beta} \le \left \langle
    \gamma_{0,1}^2 \right\rangle_{\beta},
  \qquad \left \langle(D_{0,1}
      \Gamma_{w_{\lambda}})^2 \right\rangle_{\beta} \le \left \langle \gamma_{0,1}^2 \right\rangle_{\beta}.
\end{equation}
The standard Kipnis-Varadhan argument (\cite{KV}, \cite{KLO}
chapter 1) then gives
\begin{equation*}
  \lim_{\lambda\to 0} \lambda \ll w_{\lambda}, w_{\lambda} \gg_{\beta} = 0.
\end{equation*}
It also follows form the same argument (\cite{KV}, \cite{KLO}
Chapter 1) that  $D_{0,1} \Gamma_{w_{\lambda}}$ converges strongly in
$L^2 (\rho_{\beta})$ 
to a limit that we denote with $\eta$ and that satisfies the relation
\begin{equation*}
  \left \langle \eta^2 \right\rangle_{\beta}= - \left \langle
    \gamma_{0,1} \eta \right \rangle_{\beta}  .
\end{equation*}

Before continuing we need a small technical Lemma.
\begin{lem}\label{lem:functional-extension}
Consider the linear functional defined on smooth local functions $g$ of the energies by
\begin{equation}
  \label{eq:23}
 \ell_\ve(g)= \ll j_{0,1} ,  \mc T_0 G g \gg_{\beta,\ve} .
\end{equation}
If $\gamma^{-1}\in L^2(\mu_{\beta,\ve})$, then $\ell_\ve$ can be continuously extended to the domain of $\cL$.
\end{lem}
\begin{proof}
To start with, note that $\ell_\ve$ is well defined since $G g$ is local and in $\mc H^a_\ve$ for $\ve>0$. Moreover, ${\mc T}_0 Gg$ is a local function. By using (\ref{eq:GL}) we compute
\begin{equation*}
  \begin{split}
    \ll j_{0,1} ,  \mc T_0 G g \gg_{\beta,\ve}&= \sum_z \langle
    j_{z,z+1} \, , \, {\mc T}_0 G g \rangle_{\beta,\ve}\\ 
&    =  \sum_{z,x}  \langle j_{z,z+1} [{\mc T}_0 j_{x,x+1}] \, (\partial_{e_{x+1}}
     - \partial_{e_{x}}) g \rangle_{\beta,\ve} \\
    & = \sum_{y=-1,0,1} \sum_{x} \langle j_{x+y,x+y+1} [{\mc T}_0 j_{x,x+1}] \, 
     (\partial_{e_{x+1}} - \partial_{e_{x}}) g \rangle_{\beta,\ve} \\
   &  = \sum_{y=-1,0,1}   \langle j_{y,y+1} [{\mc T}_0 j_{0,1}] 
     (\partial_{e_1} - \partial_{e_{0}}) \Gamma_g \rangle_{\beta,\ve}\\
    & = \sum_{y=-1,0,1}  \left\langle \frac{j_{y,y+1} {\mc T}_0
        j_{0,1}}{\gamma_{0,1}}  
     \; D_{0,1} \Gamma_g \right\rangle_{\beta,\ve}.
  \end{split}
\end{equation*}
That can be bounded by 
\begin{equation*}
  \sum_{y=-1,0,1}  \langle \,  \langle j_{y,y+1} \, {\mc T}_0 j_{0,1}|\mb
    e\rangle_{\beta,\ve}^2 \; \gamma_{0,1}^{-2}\, \rangle_{\beta,\ve}^{1/2}   \;
     \langle (D_{0,1} \Gamma_g)^2 \rangle_{\beta,\ve}^{1/2}.
\end{equation*}
Since we assumed here that there is no pinning potential, we have that
energies are function of the only velocities and
\[
\left< (D_{0,1}  \Gamma_g)^2\right>_{\beta,\ve} =  \left< (D_{0,1}
   \Gamma_g)^2\right>_{\beta,0}=\ll g,(-\cL)g\gg_{\beta,0}.
\]   
Since we assumed that the smooth local function are a core for $\cL$ (see Remark \ref{rem:core}), it follows that
$\ell_\ve$ can be extended to any non-local function of the energies
$g$ that belongs to the domain of $\cL$.
\end{proof}

Notice that, using \eqref{eq:steponeb}, Proposition \ref{prop:31} and Lemma \ref{lem:functional-extension}
\begin{equation}
\begin{split}
\ll j_{0,1}, v_{\lambda,1} \gg_{\beta,\ve}
&= \ll j_{0,1}, {\mc P}^a v_{\lambda,1} \gg_{\beta,\ve } \\
& =  \ll j_{0,1}, \cT_0 j_{0,1} \gg_{\beta,\ve} 
+ \ll j_{0,1},\cT_0 G w_{\lambda} \gg_{\beta,\ve}. 
\end{split}
\end{equation}
Thus, by the first line of \eqref{eq:3}, $\kappa_2 (\noise)$ is given by
\begin{equation}
\label{eq:k2345}
\beta^{-2}\kappa_2 (\noise)= \lim_{\ve \to 0} \lim_{\lambda \to 0} \left\{  \ll j_{0,1}, {\mc T}_0 j_{0,1} \gg_{\beta,\ve} + \ll j_{0,1}, \mc T_0 G w_{\lambda} \gg_{\beta,\ve} \right\}.
\end{equation}

Next, we will use the above formula for $\kappa_2 (\noise)$ for a rigorous study of the lowest order term. The first step consists to make sense, for $\ve>0$ and $\lambda>0$ fixed, of the two terms involved. 


To compute the limit of the second term of the RHS of (\ref{eq:k2345})
it is convenient to use the following Lemma. Observe first that if $g$
is a smooth local function depending only on the energies then
$j_{0,1} g$ is an antisymmetric function and $\mc T_0 (j_{0,1} g) =
(\mc T_0 j_{0,1}) \, g$. 


\begin{lem}\label{lem:orsomethinglikethat}
There exists a constant $C, \ve_0>0$ such that for any $0\le \ve<\ve_0$ and for each function $g$ of the energies $\{e_x\}$ we have
\begin{equation}
\left| \, \langle j_{-1,0}, \mc T_0 j_{0,1}
g\rangle_{\beta,\ve} \right| \le C \, \ve\, 
\sqrt{\langle g^2 \rangle_{\beta,0}} . 
\label{eq:20} 
\end{equation}
and
\begin{equation}
\left| \langle j_{0,1}, \mc T_0 j_{0,1} g\rangle_{\beta,\ve} - \langle
  \gamma_{0,1}^2 g \rangle_{\beta,\ve} \right| \le C \, \ve\,
\sqrt{\langle g^2 \rangle_{\beta,0}} . 
\label{eq:21}
\end{equation}
\end{lem}

\begin{proof}
Note that the adjoint of $L_0$ with respect to $\langle \cdot
\rangle_{\beta,\ve}$ is given by 
\[
L_0'f=L_0^*f -\ve\beta\sum_x(p_{x+1}-p_x)\nabla V(q_{x+1}-q_x) f,
\]
where $L_0^*$ is the adjoint with respect to $\langle \cdot
\rangle_{\beta,0}$. Notice that the second term of the
  above expression contains a formal infinite sum, i.e. $L_0'f$ is a
  distribution well defined against any local function.

Then, observing that 
$$
\langle p_{-1}\cdot\nabla V(q_0-q_{-1}),\mc T_0 j_{0,1} g
\rangle_{\beta,\ve} = 0 
$$
we have
\begin{equation}
\label{eq:lokp}
\begin{split}
\langle j_{-1,0}\, , & \, {\mc T}_0 j_{0,1} g \rangle_{\beta,\ve}=\frac
12\langle (p_0-p_{-1}) \cdot \nabla V(q_0-q_{-1}), \mc T_0 j_{0,1}
g\rangle_{\beta,\ve}\\ 
&=\frac 12\langle L_0^* V(q_0-q_{-1}), \mc T_0 j_{0,1} g\rangle_{\beta,\ve}\\
&=\frac 12\langle L_0' V(q_0-q_{-1}),\mc T_0 j_{0,1}g \rangle_{\beta,\ve}\\
& +\frac 12\ve\beta\sum_x\langle(p_{x+1}-p_x) \cdot \nabla V(q_{x+1}-q_x)
V(q_0-q_{-1}),\mc T_0 j_{0,1}g \rangle_{\beta,\ve}.
\end{split}
\end{equation}

In the above equation, the term $\langle L_0' V(q_0-q_{-1}),\mc T_0 j_{0,1}g \rangle_{\beta,\ve}$ is well defined as well as the second one  because $V(q_0 -q_1)$ is a local function and
  $$\sum_x\langle(p_{x+1}-p_x) \cdot \nabla V(q_{x+1}-q_x)
  V(q_0-q_{-1}),\mc T_0 j_{0,1}g \rangle_{\beta,\ve}$$
is a finite sum, in fact equal to
  \begin{equation}
  \label{eq:secondterm}
    \sum_{x=-1,0,1} \langle \langle [{\mc T}_0 j_{0,1}] (p_{x+1}-p_x)|\bq,
    \mathbf e\rangle_{\beta,\ve} \cdot \nabla V(q_{x+1}-q_x) V(q_0-q_{-1})g
    \rangle_{\beta,\ve} 
  \end{equation}
since $\langle [{\mc T}_0 j_{0,1}] (p_{x+1}-p_x)|\bq,
    \mathbf e\rangle_{\beta,\ve} = 0$ if $x \neq -1,0,1$.

The first term of the RHS of the last equality in (\ref{eq:lokp}) is equal to
\begin{equation}
\begin{split}
&\langle L_0' V(q_0-q_{-1}), {\mc T}_0 j_{0,1}g
\rangle_{\beta,\ve}\\
&=\langle V(q_0-q_{-1}), L_0 \, ({\mc T}_0 j_{0,1}g) 
\rangle_{\beta,\ve} \\
&=\langle V(q_0-q_{-1})\, , \, [ L_0 ({\mc T}_0 j_{0,1}) ] \, g) \rangle_{\beta,\ve} \\
&=\langle f \, , \, [ L_0 ({\mc T}_0 j_{0,1}) ] \rangle_{\beta,\ve}
\end{split}
\end{equation}
where $f=g V(q_0 -q_{-1})$ is a symmetric function of the velocities. We claim now that
\begin{equation}
\label{eq:claim00}
\langle f \, , \, [ L_0 ({\mc T}_0 j_{0,1}) ] \rangle_{\beta,\ve}=0.
\end{equation}
Indeed, let $\nu>0$ and write
\begin{equation*}
j_{0,1}= (\nu -L_0){\mc P}^a ((\nu -L_0)^{-1} j_{0,1} + (\nu -L_0){\mc P}^s ((\nu -L_0)^{-1} j_{0,1}
\end{equation*}
and take the scalar product of both sides with $f$ to get
\begin{equation*}
0= \langle (\nu -L_0){\mc P}^a (\nu -L_0)^{-1} j_{0,1} \, , \, f \rangle_{\beta,\ve}+ \langle (\nu -L_0){\mc P}^s (\nu -L_0)^{-1} j_{0,1} \, , \, f \rangle_{\beta,\ve}.
\end{equation*}
The first term of the RHS goes to $\langle f \, , \, [- L_0 ({\mc T}_0 j_{0,1}) ] \rangle_{\beta,\ve}$ as $\nu \to 0$ and the second term is equal to 
\begin{equation*}
\nu \langle {\mc P}^s (\nu -L_0)^{-1} j_{0,1} \, , \, f \rangle_{\beta,\ve} \; -\; \langle L_0  \, {\mc P}^s (\nu -L_0)^{-1} j_{0,1} \, , \, f \rangle_{\beta,\ve}.
\end{equation*}
It is easy to see that the second term of the previous expression is equal to $0$. This is because $L_0 =A + \noise S$, $A$ maps a symmetric function of the $p_x$'s into an antisymmetric function of the $p_x$'s, $S$ is symmetric w.r.t. $\mu_{\beta,\ve}$ and $Sf=0$. 

Thus, to prove the claim (\ref{eq:claim00}) we are reduced to show that
\begin{equation}
\lim_{\nu \to 0} \nu \, \langle {\mc P}^s (\nu -L_0)^{-1} j_{0,1} \, , \, f \rangle_{\beta,\ve} = \lim_{\nu \to 0} \nu \,\langle (\nu -L_0)^{-1} j_{0,1} \, , \, f \rangle_{\beta,\ve} =0.
\end{equation}
By Schwarz inequality it is sufficient to bound the $L^2 (\mu_{\beta,\ve})$ norm of the local function $u_{\nu}= (\nu -L_0)^{-1} j_{0,1}$. By definition we have
$$\nu u_{\nu} -L_0 u_\nu = j_{0,1}.$$
Since $j_{0,1}= -\tfrac{1}{2} S j_{0,1}$ a classical argument (\cite{KLO} Chapter 1) shows that
\begin{equation}
\langle u_{\nu}, u_{\nu} \rangle_{\beta,0} \le C \nu^{-1}
\end{equation} 
where $C$ is a constant independent of $\nu$. Observe that the support of the local function $u_{\nu}$ is fixed independently of $\nu$ so that there exist $K, \ve_0>0$ such that for any $0<\ve<\ve_0$, we have 
\begin{equation}
\langle u_{\nu}, u_{\nu} \rangle_{\beta,\ve} \le K \langle u_{\nu}, u_{\nu} \rangle_{\beta,0} \le CK \nu^{-1}.
\end{equation}
This concludes the proof of (\ref{eq:claim00}) and it remains only to show that  there exists $C, \ve_0>0$ independent of $g$ such that (\ref{eq:secondterm}) is bounded by 
\begin{equation}
 C\,  \ve\,  \sqrt{ \langle g^2 \rangle_{\beta,0}}, \quad \ve<\ve_0.
\end{equation}
This follows from Schwarz inequality. Therefore (\ref{eq:20}) is proved.


Proof of \eqref{eq:21} follows a similar line.
\end{proof}
Applying the above lemma with $g \equiv 1$ to \eqref{eq:k2345},
it follows 
\[
 \ll j_{0,1}, {\mc T}_0 j_{0,1} \gg_{\beta,\ve}= \left\langle \gamma_{0,1}^2 \right\rangle_{\beta,0}+\cO(\ve).
\]
We are left with the second term in \eqref{eq:k2345}. Using again Lemma \ref{lem:orsomethinglikethat} we have
\[
\begin{split}
&\sum_y \ll  j_{0,1}, {\mc T}_0 j_{y-1,y} \, (\partial_{e_y} - \partial_{e_{y-1}})
      w_{\lambda} \gg_{\beta,\ve}\\
&=\sum_{x,y}\langle j_{x-1,x}, \mc T_0  j_{y-1,y} \, 
(\partial_{e_y} - \partial_{e_{y-1}}) 
      w_{\lambda} \rangle_{\beta,\ve}\\
&=\sum_{x}\langle j_{x-1,x},\mc T_0 j_{x-1,x}\, 
(\partial_{e_x} - \partial_{e_{x-1}}) 
      w_{\lambda} \rangle_{\beta,\ve}+\cO(\ve)\\
&=\langle j_{0,1}, \mc T_0  j_{0,1} \,  (\partial_{e_1} - \partial_{e_{0}})
      \Gamma_{w_{\lambda}} \rangle_{\beta,\ve}+\cO(\ve)    \\
&= \langle \gamma_{0,1}^2  (\partial_{e_1} - \partial_{e_{0}})
      \Gamma_{w_{\lambda}} \rangle_{\beta,0} + \cO(\ve) 
=  \langle \gamma_{0,1} D_{0,1}\Gamma_{w_{\lambda}}
\rangle_{\beta,0} + \cO(\ve)
\end{split}
\]
In the second equality we used the fact that the sum over $y$ can be reduced to the sum over $y=x-1,x,x+1$ since the other terms are $0$. Observe that he remainder terms $ {\mc O} (\ve)$ are uniform in $\lambda$.

Therefore, we have that 
\begin{equation}
  \begin{split}
 \lim_{\lambda\to 0}    \ll j_{0,1}, v_{\lambda,1} \gg_{\beta,\ve} 
    &=  \left\langle \gamma_{0,1}^2 \right\rangle_{\beta,0} +
    \lim_{\lambda\to 0} \left\langle 
      \gamma_{0,1} D_{0,1} \Gamma_{w_{\lambda}}
    \right\rangle_{\beta,0} + {\mc O} (\ve)\\
    &= \left\langle \gamma_{0,1}^2 \right\rangle_{\beta,0} -
    \left\langle \eta^2 \right\rangle_{\beta,0} +  {\mc O} (\ve)\\ 
    &= \left\langle \gamma_{0,1}^2 \right\rangle_{\beta,0} +
    \left\langle \eta^2 \right\rangle_{\beta,0} + 2 \left\langle
      \gamma_{0,1} \eta \right\rangle_{\beta,0} + {\mc O} (\ve)\\
      & = \left\langle
      (\gamma_{0,1}+  \eta)^2 \right\rangle_{\beta,0} +  {\mc O} (\ve). 
  \end{split}
\end{equation}




We conclude that 
\begin{equation}
\label{eq:k2def}
  \begin{split}
    \kappa_2(\noise)
    = \beta^{2} \left\langle
      (\gamma_{0,1}+  \eta)^2 \right\rangle_{\beta,0} \; \ge \; 0. 
  \end{split}
\end{equation}

It follows from the above calculation that (recalling the notations introduced at the beginning of this section) 
  \begin{equation} 
\label{eq:gke}
\begin{split}
\beta^{-2} \kappa_2 (\noise)  = & \left\langle  \gamma^2_{0,1}
\right\rangle_{\beta} -  \ll  \alpha_{0,1}  \, , \, (-\mc L)^{-1}
\alpha_{0,1} \gg_{\beta}.   
  \end{split}
\end{equation}

The right hand side of \eqref{eq:gke} is \textbf{exactly} the
macroscopic diffusion of the energy in the autonomous stochastic dynamics
describing the evolution of $\mb e$, obtained in the weak coupling
limit \cite{DL, LO, LOS}. Thus even if 
\eqref{eq:gke} is
obtained from a formal expansion it is a mathematically well defined
object and we expect it coincides with $\lim_{\ve\to 0} \ve^{-2} \kappa(\ve,\noise)$.
 \\


So that
we have proved the following proposition. 

\begin{prop}
Assume that the pinning potential $W=0$ and that $V, \nabla V$ are uniformly bounded. Then we have
\begin{equation}
\lim_{\ve \to 0} \lim_{\lambda \to 0} \left\{  \ll j_{0,1}, {\mc T}_0
  j_{0,1} \gg_{\beta,\ve} + \ll j_{0,1}, \mc T_0 G w_{\lambda}
  \gg_{\beta,\ve} \right\} 
\end{equation} 
exists and is equal to
$$
\left\langle  \gamma^2_{0,1}
\right\rangle_{\beta} -  \ll \alpha_{0,1}  \, , \, (-\mc L)^{-1}
\alpha_{0,1} \gg_{\beta}.  
$$ 
\end{prop}

\noindent
\textbf{Remark: Lower bounds on $\kappa_2(\noise)$.}
Notice that $ \langle \gamma_{0,1}^{-1} \eta \rangle_{\beta} = 0$, so
\begin{equation*}
  1 = \langle \gamma_{0,1}^{-1} (\gamma_{0,1} + \eta) \rangle_\beta \; \le\; 
  \langle \gamma_{0,1}^{-2} \rangle_\beta^{1/2}
 \langle (\gamma_{0,1} + \eta)^2 \rangle_\beta^{1/2}
\end{equation*}
In particular using the last line of \eqref{eq:k2def}
\begin{equation}
\label{eq:k2ublb}
 \left\langle \gamma_{0,1}^{-2} \right\rangle_{\beta}
 ^{-1} \; \le \;\beta^{-2} \kappa_2 (\noise) \; \le \;  \left\langle
   \gamma_{0,1}^{2} \right\rangle_{\beta}
\end{equation}

In \cite{LO}, a single particle Hamiltonian of the form $H = p^2/2
+ W(q)$ is considered in dimension $d=2$. 
It is shown there under suitable assumptions on the potentials $V$ and
$W$ that  
the bound $\gamma^2(e_0, e_1)\ge c_- (\noise) e_0 e_1$ holds, for
small energies and $c_- (\noise) > 0$ for $\noise > 0$.  
It follows that the lower bound \eqref{eq:k2ublb} is strictly positive
as soon as 
$\noise >0$ for that system. We conjecture that this holds in general
for $\noise >0$ and we prove it for the examples of Section
\ref{sec:ex1}.

When the Hamiltonian part of the cell dynamics is given by a geodesic
flow on a manifold of negative curvature, the lower bound in
\eqref{eq:k2ublb}  is strictly positive even without
the noise ($\noise=0$), in dimension $d\ge 3$ (\cite{DL}).

\section{The non-equilibrium stationary state}
\label{sec:remark-non-equil}

Instead of studying the energy flux via the GK, an alternative, more direct, approach is possible:
one can consider the stationary state in a finite open system with Langevin thermostats at the
boundaries having temperatures $T$ and $T+ \delta T$ respectively, \cite{BO11}. To simplify the study we assume that $d=1$.
The generator of the dynamics is then 
$$
L_{\ve, N, \delta T} =  \sum_{x=1}^N (A^{x} + \noise S^x) + \ve G +
B_{1,T+\delta T} + B_{N, T} 
$$
where $B_{1,T+\delta T}, B_{N, T}$ are the generators of the
corresponding Langevin dynamics at the boundaries:
$$
B_{x,T} = \frac{T}{2} \partial_{p_x}^2 - p_x  \partial_{p_x}
$$
and
$$
G = \sum_{x=2}^{N} \nabla V(q_x
    - q_{x-1}) (\partial_{p_{x-1}} - \partial_{p_{x}}) \ .
$$
Our goal is to compute the thermal conductivity of the stationary
state, e.g. the stationary current divided by the temperature gradient
$\delta T/N$:  
\begin{equation}
  \label{eq:10}
  \kappa_{N,T, \ve} = \lim_{\delta T\to 0} \frac N{\delta T}\ \ve
  \left< j_{0,1} \right>_{N,\delta T, \ve} ,
\end{equation}
where $<\cdot>_{N,\delta T, \ve}$ is the expectation with respect to
the stationary measure. 
To this end we are going to expand the stationary measure
in $\ve$ and $\delta T$.  

Let us reiterate once more that the following is only formal. Indeed, to simplify the presentation, we do not insist on issues that have already been treated more carefully in the previous sections.
As a preliminary step, we use as a reference measure the inhomogeneous
Gibbs distribution with 
linear profile of inverse temperature $\{\beta_x\}_{x=1,\dots,N}$,
interpolating between the two inverse temperatures by setting
$\beta_{x+1} - \beta_x \sim -\frac{\delta T}{N T^2}$. 
We will call $\bE$ the expectation with respect to such a measure,
that is 
\begin{equation}\label{eq:first-app}
\bE(f)=Z^{-1}\int e^{-\sum_{x=1}^N\beta_x e^\ve_{x}}f(q,p) dq dp,
\end{equation}
where as before $e^\ve_{x}=\frac 12 p_x^2+W(q_x)+\frac 12\ve
[V(q_{x}-q_{x-1})+V(q_{x+1}-q_{x})]$, for $x=2,\dots, N-1$ and
$e^\ve_{1}=\frac 12p_1^2+{W(q_x)}+\frac 12\ve V(q_{2}-q_{1})$, $e^\ve_{N}=\frac
12p_N^2+{W(q_x)}+\frac 12\ve V(q_{N}-q_{N-1})$.\footnote{ Since we will compute
  a correction of order one, the correction to the local energies
  does not really matter.} To keep consistency with previous notations, we will
use $e_x$ to designate $e^0_{x}$, the internal energy of the isolated cell. 

The corresponding adjoint operator is  
$$
L_{\ve, N, \delta T}^* =  \sum_{x=1}^N (- A^{x} + \noise S^x) - \ve G + \ve
\sum_{x=1}^{N-1} (\beta_{x+1} - \beta_x) j_{x,x+1} +
B_{1,T+\delta T} + B_{N, T} .
$$
We assume that there exists a unique stationary probability distribution with smooth density. 
The existence and uniqueness of such a probability measure still remains an open problem for most of the dynamics that appear in this work, 
though for some models, proofs can be found in \cite{BO11} (see also
\cite{Luc}).  
For certain choices of the local dynamics $L_0$ and interaction $V$,
the smoothness  of the density  follows by applying results of \cite{EH}, \cite{Car}. 

Let $f_{\ve, N, \delta T}$ be the density of this stationary measure with
respect to this inhomogeneous Gibbs measure, i.e. the solution of 
$$
L_{\ve, N, \delta T}^* f_{\ve, N, \delta T} = 0, \quad f_{\ve, N, \delta T}\ge
0. 
$$
It is observed that the energy of the particles at the boundary sites $(x=1,N)$ is not conserved due to the action of the reservoirs. 
So, if  $\mb e= \{e_2,\dots,e_{N-1}\}$, it is convenient
to define the projector\footnote{ Note that this projector is different from the one used in Section \ref{sec:form-expans-green}.} 
\[
\Pi f (e_2, \dots, e_{N-1}) =  \bE(f\;|\; \mb e) \ .
\]
Also let $B=B_{1,T} + B_{N, T} $ and 
$$
\coJ= \frac 1{NT^2}\sum_{x=1}^{N-1}  j_{x,x+1}.
$$

We expand the stationary measure as follows
\begin{equation}
  \label{eq:5}
  f_{\ve, N, \delta T} =1+\delta T\left[ w_0 + \sum_{n\ge 1} \left(v_n
      + w_n \right) \ve^n\right]+\cO((\delta T)^2) 
\end{equation}
where $\Pi w_n=w_n$ and $\Pi v_n =0$. Next, it is convenient to set 
\[
L_B=\sum_{x=1}^N (A^{x} + \noise S^x)+B= L_{0,N,0}.
\]
Note that\footnote{ Here the adjoint is taken with respect to all the
  measures $\bE(\cdot\;|\; {\mb e})$.}
\[
L_B^*=\sum_{x=1}^N (-A^{x} + \noise S^x)+B= L_{0,N,0}^*
\]
and that
$L_B^*\Pi=\Pi L_B^*=0$.
Since $L_{\ve, N, \delta T}^*{\bf 1}= \coJ \delta T + \cO((\delta T)^2)$, if we
compute at the first order in $\delta T$ we have 
\[
-\ve \coJ-\ve Gw_0+\sum_{n\geq 1}\ve^n\left\{L_B^* v_n-\ve Gv_n-\ve G w_n\right\}=0.
\]
From the above it follows
\begin{equation}\label{eq:details}
\begin{split}
 L_B^* v_1 &=\coJ +G w_0 \\
L_B^* v_{n+1} &= G w_n + Gv_n \quad\text{ for }n>0 .
\end{split}
\end{equation}
We proceed similarly to Section \ref{sec:form-expans-green}, starting from $n=1$.
Since, again, $\Pi G\Pi=0$, applying $\Pi$ to the second equation in \eqref{eq:details}, when $n=1$, we obtain\footnote{To simplify the presentation we assume that $L_B^{-1}$ is well defined. For a more rigorous argument it suffices to use the analogous of Proposition \ref{prop:31}.}
\begin{equation} \label{eq:ufff}
\begin{split}
 v_1 &=(L_B^*)^{-1}\left[\coJ +G w_0\right]\\
\Pi Gv_1&=0.
\end{split}
\end{equation}
We can then multiply the first equation by $\Pi G$ and define the operator
\begin{equation}\label{eq:cLB}
\cL_B=\Pi G (L_{B}^*)^{-1}G\Pi.
\end{equation}
This readily implies that
\begin{equation} \label{eq:details-more}
\begin{split}
  w_0&= \cL_B^{-1} \Pi G (-L_B^*)^{-1}\coJ\\
 v_1 &=(L_B^*)^{-1}\left[\coJ +G w_0\right],
\end{split}
\end{equation}
solve \eqref{eq:ufff}. Next, we consider $n>1$. 
Applying $\Pi$ to the second of the \eqref{eq:details} we have that $v_n$ must satisfy $\Pi G v_n=0$ for all $n\in\bN$ (recall that $\Pi G \Pi = \Pi L_B^* = 0$). 
As we have seen, this is indeed the case for $n=1$. Assume it and try for $n>1$. Then, we can apply $(L_B^*)^{-1}$ and obtain
\[
v_{n+1}=(L_B^*)^{-1}\left[Gw_n+Gv_n\right].
\]
Multiplying for $\Pi G$ yields, for $n\geq 1$,
\begin{equation} \label{eq:7-1}
  \begin{split}
&w_n= \cL_B^{-1} \Pi G (-L_B^*)^{-1}Gv_n  \\
&  v_{n+1}=(L_B^*)^{-1}\left[ Gw_n+ G v_n\right] .
 \end{split}
\end{equation}
Note that 
\[
\Pi Gv_{n+1}=\Pi G(L_B^*)^{-1}\left[ Gw_n+ G v_n\right]=\cL_Bw_n+\Pi G(L_B^*)^{-1} G v_n=0
\]
as needed.

Next, we want to compute how $\cL_B$ acts on the space of function $\{f\;:\; \Pi f=f\}$.
\[
\begin{split}
Gf&=\sum_{x=2}^N \nabla V(q_x-q_{x-1})(p_x\partial_{e_x}-p_{x-1}\partial_{e_{x-1}})f\\
&=\sum_{x=2}^N j_{x-1,x}(\partial_{e_x}-\partial_{e_{x-1}})f
-\frac 12 \sum_{x=2}^N  [L_B^* V(q_x-q_{x-1})](\partial_{e_x}+\partial_{e_{x-1}})f.
\end{split}
\]
Thus, given two function of the energies $f(e_2, \dots, e_{N-1})$ and
$g(e_2, \dots, e_{N-1})$, we have\footnote{ By $\bE_\beta$ we mean the measure \eqref{eq:first-app} with $\delta T=0$.}
\begin{equation}\label{eq:dichlet}
\begin{split}
\bE_\beta(g \cL_B f)& = \bE_\beta(g \Pi G (L_B^*)^{-1} G\Pi f)\\
&= \sum_{x=2}^N\bE_\beta(g  G
(L_B^*)^{-1}j_{x-1,x} (\partial_{e_x}-\partial_{e_{x-1}})f),
\end{split}
\end{equation}
where we have used the antisymmetry in $p$ of the measure.
Also, taking the adjoint with respect to $\bE_\beta$ yields
\begin{equation}\label{eq:G-adjoint2}
G^*=-G+\beta\sum_{x=2}^N \nabla V(q_x-q_{x-1})(p_x-p_{x-1})=-G+\beta L_B V.
\end{equation}
Inserting the above in \eqref{eq:dichlet} 
and using again the antisymmetry in $p$ we have
\[
-\bE_\beta(g\cL_B f)=\frac 14\sum_{x,y =2}^N \bE_\beta \left(
  (j_{y,y-1}(\partial_{e_y}-\partial_{e_{y-1}})g \cdot
 (- L_B^*)^{-1}j_{x, x-1}(\partial_{e_x}-\partial_{e_{x-1}})f \right).
\]
Finally, we have
\begin{equation*}
  \begin{split}
    \gamma(e_{x-1},e_x)^2\delta_{xy} &= \bE_\beta( j_{y-1,y}
    (-L_B^*)^{-1}j_{x-1,x}\;|\;{\mb e}) \qquad x\neq 2, N\\
    \gamma(e_1, e_2)^2\delta_{y,2} &= \bE_\beta( j_{y-1,y}
    (-L_B^*)^{-1}j_{1,2}\;|\;{\mb e}) \\
    \gamma(e_{N-1}, e_N)^2\delta_{y,N} &= \bE_\beta( j_{y-1,y}
    (-L_B^*)^{-1}j_{N-1,N}\;|\;{\mb e})
  \end{split}
\end{equation*}

Thus
\[
\begin{split}
-\bE_\beta(g\cL_B f)&= \frac
14\sum_{x=2}^{N}\bE_\beta(\gamma^2(e_{x-1},e_{x})(\partial_{e_x}-\partial_{e_{x-1}})g\cdot(\partial_{e_x}-\partial_{e_{x-1}})f)\\  
&= \frac 14\sum_{x=3}^{N-1}\bE_\beta(\gamma^2(e_{x-1},e_{x})(\partial_{e_x}-\partial_{e_{x-1}})g
\cdot(\partial_{e_x}-\partial_{e_{x-1}})f)\\
&\quad + \frac 14
\bE_\beta(\gamma^2(e_1,e_2)\partial_{e_2}g\cdot \partial_{e_2}f)\\
& \quad+ \frac 14\bE_\beta(\gamma^2(e_{N-1}, e_N)\partial_{e_{N-1}}g\cdot\partial_{e_{N-1}}f)
\end{split}
\]
which shows that $\cL_B$ is the operator that one would expect in
\cite{LO,DL} when adding the appropriate boundary terms.

Also note that, in analogy with Proposition \ref{prop:Lweak} and \eqref{eq:res}, 
\begin{equation}\label{eq:alpha-B}
\begin{split}
\alpha(e_x,e_{x+1})&=\Pi G(-L^*_B)^{-1}j_{x,x+1}=\Pi G^*(-L_B)^{-1}j_{x,x+1}\\
&=D^*_{x,x+1}\gamma(e_x,e_{x+1}).
\end{split}
\end{equation}
We can, at last, compute the current:
\[
\bE(f_{\ve,N,\delta T} j_{0,1})=\delta T\,\bE(\{w_0+\ve( v_1+w_1)\}j_{0,1})+\cO(\ve^2\delta T+(\delta T)^2).
\]
Thus, setting 
\[
{\bf j}_{0,N}=\lim_{\delta T\to 0}\frac 1{\delta T}\bE(f_{\ve,N,\delta T} j_{0,1})
\]
we have formally\footnote{Note that here, since we are dealing with finite systems, there is no problems in expanding the Gibbs measure in $\delta T$ and $\ve$.} 
\[
\begin{split}
{\bf j}_{0,N}&=\ve \bE_\beta( v_1j_{0,1})+\cO(\ve^2)\\
&=\frac {\ve}{N T^2}\bE_\beta(\gamma(e_{0},e_{1})^2)
+\ve\bE_\beta(\coJ\cdot (L_B^*)^{-1}G w_0)+\cO(\ve^2)\\
&=\frac {\ve\beta^2}{N }\bE_\beta(\gamma(e_{0},e_{1})^2)\\
&\qquad + 
\frac{\ve\beta^2}{N} \sum_{x=1}^{N-1}\bE_\beta( \alpha(e_{x-1}, e_{x}) [(-\cL_B)^{-1}
    \alpha (e_0, e_1)] )+\cO_N(\ve^2)
\end{split}
\]

Therefore, formally, the limit 
$$
{ \lim_{\ve \to 0}\; \lim_{N\to \infty} \; \lim_{\delta T \to 0} \; \frac{1}{\ve^2} \kappa_{N,T,\ve}}
\; = \;
\lim_{\ve \to 0}\; \lim_{N\to \infty} \frac{N}{\ve}  {\bf j}_{0,N},
$$ 
{with $ \kappa_{N,T,\ve}$ given by \eqref{eq:10}},
yields the  formula for $\kappa_2(\noise)$ in agreement with the
Green-Kubo formula expansion of Section \ref{sec:green-kubo-formula2}. 

\section{Behavior of $\kappa_2(\noise)$ in the limit
  $\noise \to 0$ for some model systems} 
\label{sec:ex1}

We now study the behavior of $\kappa_2 (\noise)$ in the
deterministic limit $\noise \rightarrow 0$. This limit is singular,
since the operator $\mathcal L = \Pi G (-L_0)^{-1}G \Pi$
formally vanishes at $\noise = 0$ for the 
whole class of systems considered in this work:  
both operators $(-L_0)^{-1}$ and $G$ exchange symmetric and
antisymmetric functions under the operation $p\rightarrow -p$,  
while $\Pi$ annihilates antisymmetric functions. It is therefore important to analyse some particular cases
in more detail. When the dynamics of individual cells is chaotic, the operator $(-L_0)^{-1}$ can be defined at $\noise = 0$ only on appropriate spaces of distributions, \cite{Liv}, in which the operator induced by the map $p\rightarrow -p$ is unbounded. Thus the above formal argument does not hold and the operator $\mathcal L$ does not trivially vanishes for $\noise=0$, \cite{DL}.

Another way to bypass the above problems is by looking at integrable isolated dynamics for which fairly explicit computations can be carried out.   
Here we consider three such examples. 
In all these cases, the uncoupled cells are one-dimensional and the stochasticity is the random velocity flip with rate $\noise^{-1}$. 

\emph{1. Anharmonic oscillators.} 
It is a common belief, based on extensive numerical simulation,
that the transport of energy in anharmonic
one-dimensional pinned chains is diffusive \cite{LLP}\cite{dhar} 
(see also 
\cite{GAGI} for physical approaches passing through a weak coupling limit).  
However, to our knowledge, there are no rigorous mathematical
arguments supporting this.  
We show here that $\limsup_{\noise\to 0}\kappa_2(\noise) <\infty$
for one-dimensional oscillators with rather generic pinning
potentials $W$ and interaction $V$. 

We consider the Hamiltonian
\eqref{Hamiltonian one dimensional} below which allows for an
explicit description. The fact that as
$\noise \rightarrow 0$, $\kappa_2 (\noise)$ does not diverge
results from averaging 
oscillations in the uncoupled cells, and not from decay of correlations as it would be the
case for a chaotic dynamics. The control of the time integrated
current-current correlations in the limit $\noise \rightarrow 0$ is
possible if resonances between near atoms occur with small probability
in the Gibbs state. This condition is violated if the pinning $W$ is
harmonic, but is otherwise typically satisfied.  

\emph{2. Disordered oscillators and rotors.}
We next consider in more details two examples of chains of one
dimensional systems that display a similar structure: the disordered
harmonic chain and the rotor model. 
In each case, the atoms are
one-dimensional systems, so that, when 
both noise and coupling are removed, the full dynamics becomes again
integrable.  Moreover, then, neighboring particles typically oscillate
at different frequencies. For these two examples, we are able to give
explicit formulas for the weak coupling operator ${\mc L}$ (see
Proposition \ref{prop:dc} and Proposition \ref{prop: weak coupling rotors}).  
 
In the absence of noise ($\noise = 0$), the disordered chain is well
known to be a perfect insulator: $\kappa = 0$ \cite{CBFH},  
while it is conjectured that the conductivity of the rotor chain is
finite and positive \cite{LLP},  but decays faster than any power law
in $\varepsilon$ as $\varepsilon \rightarrow 0$ \cite{WDFH}.  Thus in
these two cases it is expected that the conductivity of the
deterministic system $\kappa (\varepsilon,\noise)$ has no expansion in powers of $\varepsilon$.  
%
What we are actually able to prove is that for the disordered harmonic chain
$$\lim_{\noise \to 0} \kappa_2 (\noise)=0.$$
We also show in Subsection \ref{subsec: rotor chain} that for the rotor chain $\limsup_{\noise \to 0} \kappa_2 (\noise) <+\infty$, extending the conclusions of
Proposition \ref{prop:ubgeneric} to this case.


\subsection{Upper bound on the conductivity for pinned anharmonic oscillators.}
\label{sec:ex2}

Let
\begin{equation}\label{Hamiltonian one dimensional}
  \begin{split}
    H(q,p) \; = \; \sum_x \frac{p_x^2}{2} + W(q_x) + \ve V(q_{x+1} -
    q_x) \\
    \; = \; \sum_x H_0 (q_x,p_x) + \ve V(q_{x+1} - q_x)
  \end{split}
\end{equation}
with $(q_x,p_x) \in \RR^2$.
The potential $W$ is assumed to be smooth, strictly convex, except
possibly at the origin, and symmetric
.  
The potential $V$ is also taken smooth, symmetric, bounded below, and
of polynomial growth, always satisfying the requirement that
$\mu_{\beta,\ve}$ is analytic in $\ve$ for small $\ve$. 
To make things simple and concrete, we will actually focus on $W$ given by
\begin{equation}\label{special case W}
W (q) \; = \; \frac{|q|^r}{r}, \qquad r > 2. 
\end{equation}

\begin{prop}
\label{prop:ubgeneric}
Let $W$ be given by \eqref{special case W} for some $r > 2$. 
Then, with the assumptions on $V$ given after \eqref{Hamiltonian one
  dimensional}, $\limsup \kappa_2 (\noise) < + \infty$. 
\end{prop}

\begin{proof}
Because of its length, the proof as well as the needed introductory material are postponed to Appendix \ref{sec:proofubgeneric}.
\end{proof}

\noindent
\textbf{Remark.}
In Proposition \ref{prop:ubgeneric}, we have limited ourselves to a
case leading to rather clean computations.  
A closer look at the proof in appendix \ref{sec:proofubgeneric} shows
that our hypotheses are too restrictive:  
what is important is that the map $\omega (I)$, giving the frequency
of oscillation as a function of the action, can be inverted.  
The main advantage of taking the $W$ given by \eqref{special case W},
is that this can be done explicitly.  

It then arises as a natural question whether the proof could
be further generalized to cases where $\omega (I)$ is invertible
everywhere but on a finite or countable number of points.  
This would for example be the case if we consider the pinning
potential $W(q) = q^2 + a \cos (q)$ for some small enough constant $a>
0$.  
This is unfortunatly not the case, as some logarithmic divergence in
$\noise$ shows up in the limit $\noise \rightarrow 0$, if one just
tries to mimic the proof of Proposition \ref{prop:ubgeneric}.  
Unless the system posseses some hidden symmetry, this in fact means
that $\langle \gamma^2(e_0,e_1)\rangle_{\beta,0}$ diverges
logarithmically in the deterministic limit.  
This however does not necessarily imply that $\kappa_2$ itself will
diverge in this limit, as the term $\langle \eta^2 \rangle_{\beta,0}$ in
\eqref{eq:k2def} can compensate this divergence.  
This is in fact what is expected to happen.

\subsection{The disordered harmonic chain.}\label{subsec: disordered harmonic chain}
The hamiltonian part of the generator is now given by 
\begin{equation}\label{Generator disordered harmonic chain}
A \; = \; \sum_{x} p_x \partial_{q_x} - \omega_x^2 q_x \partial_{p_x}, 
\qquad 
G \; = \; \sum_x (q_{x-1} - 2q_x + q_{x+1}) \partial_{p_x}, 
\end{equation}
where $\omega_x^2$ are random, independent and identically distributed squared frequencies, that satisfy the bound $c^{-1} \le \omega_x^2 \le c$, for some constant $c > 0$.
The internal energy is given by $e_x = p_x^2/2 + \omega_x^2 q_x^2/2$, 
while for $\varepsilon \ge 0$ the energy flux $\varepsilon j_{x,x+1}$ between two adjacent oscillators is given by 
\begin{equation*}
\ve \, j_{x,x+1} \; = - \ve \; \frac{p_x + p_{x+1}}{2} (q_{x+1} - q_x).
\end{equation*}

\begin{lem}
Let $x, y\in \bZ$. 
A solution $\psi_{x,y}$ to the equation
\begin{equation*}
-L_0 \psi_{x,y} \; = \; q_x p_y
\end{equation*}
is given by 
\begin{equation*}
\psi_{x,y} = 
\frac{
4 \noise \left(\omega_x^2 q_x p_y 
- \omega_y^2 q_y p_x \right)
+ (\omega_x^2 - \omega_y^2) p_x p_y 
+ \big((\omega_x^2 - \omega_y^2 -8 \noise^2)\omega_y^2\big) q_x q_y} 
{\Delta(x,y)}
\end{equation*}
with
\begin{equation}\label{Delta Definition}
\Delta_{x,y} 
\; = \; 
8 \noise^2 (\omega_x^2 + \omega_y^2) +  (\omega_x^2 - \omega_y^2 )^2.
\end{equation}
\end{lem}

\begin{proof}
This follows by a direct computation. 
\end{proof}

This lemma allows us to give an explicit form of the operator $\mathcal L = \Pi G (L_0)^{-1} G \Pi$. 
We know that $\mathcal L$ is the generator of a Ginzburg-Landau dynamics. 

\begin{prop}
\label{prop:dc}
Let $\mathcal L = \Pi G (L_0)^{-1}G\Pi$. Then
\begin{align}
\rho_\beta (\mathbf e)
\; &= \;
\prod_x \Big( \beta\ed^{-\beta e_x} \Big), \\
\gamma^2 (e_x,e_{x+1})
\; &= \;
\frac{4\noise}{\Delta_{x,x+1}} e_x e_{x+1}, \\
\alpha (e_x,e_{x+1})
\; &= \; 
\frac{8\noise}{\Delta_{x,x+1}} (e_x - e_{x+1}) .
\end{align}
\end{prop}

\begin{proof}
To obtain the expression for the invariant measure, let us take an 
 $f$ that depends only on $e_x = (p_x^2 + \omega_x^2 q_x^2)/2$, and
 let us compute 
\begin{align*}
\langle f \rangle_{\beta,0}
\; &= \; 
Z_x(\beta)^{-1}
\int_{\bR^2} f \Big( \frac{p_x^2 + \omega_x^2 q_x^2}{2} \Big) \, \ed^{-\beta (p_x^2 + \omega^2_x q_x^2)/2} \, \dd q_x \dd p_x \\
\; & \sim \; \int_0^\infty f(e) \, \ed^{-\beta e} \dd e
\end{align*}
from which the expression for $\rho_\beta$ follows\footnote{
Here and in the sequel, we use $a \sim b$ to say that there exist two postive constants $C_1,C_2$ such that $a \le C_1 b$ and  $b \le C_2 a$. }.

Next we have that 
\begin{align*}
\gamma^2 (e_x,e_{x+1})
\; &= \;
\Pi \big( j_{x,x+1} (-L_0)^{-1} j_{x,x+1} \big) \\
& = \; 
\frac{1}{4} \Pi
\Big( q_{x+1}p_x - q_x p_{x+1} + q_{x+1}p_{x+1} - q_xp_x \Big) \\
&\hspace{1.5cm}
\Big( \psi_{x+1,x} - \psi_{x,x+1} + \psi_{x+1,x+1} - \psi_{x,x} \Big) \\
& = \;
\frac{2\noise}{\Delta_{x,x+1}} \Pi \Big( 
\omega_{x+1}^2 q_{x+1}^2 p_x^2 + \omega_x^2 q_x^2 p_{x+1}^2
\Big)
\end{align*}
where we have used the fact that odd powers of
$q_x,p_x,q_{x+1},p_{x+1}$ are annihilated by the projection $\Pi$.  
Using then polar coordinates
\begin{equation*}
\frac{\omega_x q_x}{\sqrt 2} \; = \: \sqrt{e_x} \cos \theta_x, 
\quad 
\frac{p_x}{\sqrt 2} \; = \; \sqrt{e_x} \sin \theta_x, 
\end{equation*} 
it is computed that both
\begin{equation*}
\Pi (p_x^2) \; = \; \Pi (\omega_x^2 q_x^2) \; = \; \frac{1}{2\pi} \int_0^{2\pi} 2e \sin^2 \theta_x \, \dd \theta_x \; = \; e_x.
\end{equation*}
This yields the announced expression for $\gamma^2 (e_x,e_{x+1})$. 

The current $\alpha (e_x,e_{x+1})$ follows using \eqref{eq:alpha}. 
\end{proof}

\begin{cor}\label{cor: dhc corol}
For $\noise > 0$, we have that a.s. in $\omega$
$$\kappa_2 (\noise)= \cfrac{8\noise\beta^2}{\langle \Delta_{0,1} (\noise) \rangle_*}>0$$
where $\langle \cdot \rangle_*$ represents the average with respect to
the realizations of the disorder. In particular, a.s. in $\omega$,  
$$\lim_{\noise \to 0} \kappa_2 (\noise)=0.$$ 
\end{cor}

\begin{proof}
The proof is given in Appendix \ref{Appendix proof dhc}.
\end{proof}

\subsection{The rotor chain.}\label{subsec: rotor chain}
The Hamiltonian part of the dynamics is given by
\begin{equation}\label{Generator rotor chain}
A \; = \; \sum_x p_x \partial_{q_x}, 
\qquad 
G \; = \;  \sum_{x} \left[ \sin(q_{x-1} -q_x) - \sin (q_x -q_{x+1}) \right] \partial_{p_x},
\end{equation} 
with $q_x \in \bR/{2\pi \bZ}$.
The individual energy for the uncoupled dynamics ($\varepsilon=0$) is $e_{x} = p_x^2 /2$. 
If $\varepsilon>0$, there is a flux of energy which is given by $\varepsilon j_{x,x+1}$ where
\begin{equation*}
 j_{x,x+1}= -\frac12 (p_x + p_{x+1}) \sin (q_{x+1} -q_x) .
 \end{equation*}

\begin{prop}\label{prop: weak coupling rotors}
For this system
\begin{align}
\rho (\mathbf e) 
\; & = \; 
\prod_{x} \Big( \ed^{- (U(e_x) + \beta e_x)} \sqrt{\beta/\pi} \Big)
\quad \text{with} \quad 
U (e_x) \; = \; \frac{1}{2} \log e_x,
\label{Gibbs measure rotors}\\
\gamma^2 (e_x,e_{x+1}) 
\; &= \; 
\frac{2\noise \, e_x e_{x+1}}{\Delta (e_x,e_{x+1})},
\label{gamma rotors}\\
\alpha (e_x,e_{x+1})
\; &= \; 
\frac{\noise (e_{x} - e_{x+1})}{\Delta^2(e_x,e_{x+1})} \big( \Delta (e_x,e_{x+1}) + 8 e_xe_{x+1} \big)
\label{alpha rotors}
\end{align}
with 
\begin{equation*}
\Delta (e_x,e_{x+1}) \; = \; 4 \noise^2 (e_x +e_{x+1}) + (e_{x+1} -e_x)^2.
\end{equation*}
\end{prop}

\begin{proof}
The proof is given in Appendix \ref{Appendix: proof of a lemma and a
  proposition}. 
\end{proof}

It is seen from  the above expressions that as noted earlier the
generator $\mathcal 
L$ formally vanishes as $\noise \rightarrow 0$.  
However, for $\noise$ small but positive, 
the coefficient $\gamma^2 (e_x,e_{x+1})$ can become of order $1 /\noise$ in case a resonance occurs, such that $|e_{x+1} - e_x| \le \varsigma$.
We have unfortunately not been able to decide whether, despite of this
phenomenon, the value of $\kappa_2 (\noise)$ still vanishes as $\noise
\rightarrow 0$,  
as suggested by the results in \cite{WDFH}.  

We have however a result analogous to that of Proposition \ref{prop:ubgeneric}:
\begin{prop}\label{prop:limsuprotors}
For any $\noise>0$, $\kappa_2 (\noise)$ is strictly positive and 
$$
\limsup_{\noise \rightarrow 0}\kappa_2 (\noise)< + \infty.
$$
\end{prop}

\begin{proof}
By (\ref{eq:k2ublb}) and the explicit form of $\gamma$ we have that
\begin{equation*}
\beta^{-2} \kappa_2 (\noise) \ge \langle \gamma^{-2} (e_0,e_1)
\rangle_{\beta,0}^{-1} \ge c \noise  
\end{equation*}
for a positive constant $c$ independent of $\noise$. By (\ref{eq:k2ublb}) it holds also that 
\begin{equation*}
\beta^{-2}\kappa_2 
\; \le \; 
\langle \gamma^2 (e_0,e_1) \rangle_{\beta,0}.
\end{equation*}
The function $\langle \gamma^2 (e_0,e_1) \rangle_{\beta,0}$ has the behavior
\begin{align*}
\langle \gamma^2 (e_0,e_1) \rangle_{\beta,0}
\; &\sim \;
\int_{\bR_+^2} \frac{\noise \, e_0 e_1}{4\noise^2 (e_0 + e_1) + (e_1 - e_0)^2} \ed^{-\beta (e_0 + e_1)}\frac{\dd e_0 \dd e_1}{\sqrt{e_0e_1}} \\
\; &\sim \; 
\int_{\bR^2} \frac{\noise \, x^2 y^2}{8 \noise^2 (x^2 + y^2) + (y^2-x^2)^2} \ed^{-\beta (x^2 + y^2)/2} \, \dd x \dd y \\
\; &\sim \; 
\int_0^\infty \dd r \int_0^{2\pi} \dd \theta  \, \frac{ \noise \, r^3 \cos^2 \theta \sin^2 \theta}{8 \noise^2 + r^2 (\cos^2 \theta -\sin^{2} \theta)^2}\, \ed^{-\beta r^2 /2}.
\end{align*}
In the limit $\noise \rightarrow 0$, 
only the values of $\theta$ such that $\cos^2 \theta -\sin^{2} \theta \sim 0$ contribute ($\theta \sim \pm \pi/4$ and $\theta \sim \pm 3\pi/4$), 
so that, by a Taylor expansion, 
\begin{align*}
\langle \gamma^2 (e_0,e_1) \rangle_{\beta,0}
\; &\sim \; 
 \int_0^{\infty} \dd r \, \int_{0}^1 du \,  \frac{\noise \, r^3}{8 \noise^2 + r^2 u^2}\; \ed^{-\beta r^2 /2}\\
\; &\sim \;
\int_0^\infty r^2 \ed^{-\beta r^2/2} \left( \int_0^1 \frac{\noise/r}{8 (\noise/r)^2 + u^2} \, \dd u \right) \, \dd r \\
\; &\sim \;  1 \quad \text{as} \quad \noise \rightarrow 0.
\end{align*}
This proves the claim.
\end{proof}

\appendix

\section{Proof of Proposition \ref{prop:ubgeneric}}
\label{sec:proofubgeneric}

To study the system at hand, it is convenient to pass to action-angle variables. 
Let $\mathrm I : \RR_+ \rightarrow \RR_+$ be defined by 
\begin{equation*}
\mathrm I (E) \; = \; \frac{1}{2\pi}\int_{A(E)} \dd q \dd p
\quad \text{with} \quad 
A(E) \; = \; \{ (q,p) \in \RR^2 : H_0(q,p) \le E \}.
\end{equation*}
Our assumptions on $W$ ensure that $\mathrm I' (E) = \dd \mathrm I/\dd E \, (E)  > 0$ for any $E > 0$.
Given $E \ge 0$, we also set
\begin{equation*}
q^* (E) \; = \; \max \{ q \in \RR : H_0(q,p) = E \text{ for some }p\in \RR \} .
\end{equation*}
Then we define the action-angle variables by
\begin{align*}
I_x \; &= \; I(q_x,p_x) \; = \; \mathrm I(H_0 (q_x,p_x)), \\
\theta_x \; & = \; \theta (q_x,p_x) \; = \; \frac{ - \mathrm{sgn} (p_x)}{\mathrm I'(H_0(q_x,p_x))} \int_{q_x}^{q^*(H_0(q_x,p_x))} \frac{\dd q'}{\sqrt{2 (H_0 (q_x,p_x) - W(q'))}}.
\end{align*}
It is checked that $(I_x,\theta_x) \in \RR_+ \times \TT$ with $\TT = \RR/(2\pi\ZZ)$.
The potential $W$ is such that this change of variable is invertible, except at origin. 
We denote by $Q$ and $P$ the inverse maps: 
\begin{equation*}
q_x = Q (I_x,\theta_x),  \qquad p_x = P(I_x,\theta_x).
\end{equation*}
The change of variables $(q_x,p_x) \leftrightarrow (I_x,\theta_x)$ is known to be a canonical change of variables. 

Let $\mathrm H_0:\R_+ \rightarrow \R_+$ be the inverse function of $\mathrm I$: $\mathrm H_0\circ \mathrm I (E) = E$ for any $E\in\R_+$.
In the action-angle variables, the Hamiltonian \eqref{Hamiltonian one dimensional} reads
\begin{equation*}
H(I,\theta) 
\; = \;
\sum_x \mathrm H_0 (I_x) + \epsilon V \big( Q (I_{x+1}, \theta_{x+1})
- Q(I_x,\theta_x) \big).  
\end{equation*}
Defining
\begin{equation*}
\omega (I_x) \; = \; \mathrm H_0' (I_x) \; = \; \dd \mathrm H_0 / \dd I_x ,
\end{equation*}
Hamilton equations read
\begin{equation*}
  \begin{split}
    \dot{I}_x \; &= \; - \epsilon \frac{\partial }{\partial \theta_x} V
    \big( Q (I_{x+1}, \theta_{x+1}) - Q(I_x,\theta_x) \big), \\
    \dot{\theta}_x \; &= \; \omega_x \, + \, \epsilon
    \frac{\partial}{\partial I_x} V \big( Q (I_{x+1}, \theta_{x+1}) -
    Q(I_x,\theta_x) \big).
  \end{split}
\end{equation*}
The current, given by \eqref{current: general formula}, has the form
\begin{equation}\label{current action angle}
j_{x,x+1} \; = \; - \frac{1}{2} \big(P(I_x,\theta_x) + P(I_{x+1}, \theta_{x+1}) \big) V'\big( Q (I_{x+1}, \theta_{x+1}) - Q(I_x,\theta_x) \big)
\end{equation}
with $V'(x) = \dd V/\dd x$.

Since we are in dimension $d=1$, the noise written in the action-angle
coordinates is given by  
\begin{equation}\label{noise action angle}
S f (I , \theta) \; = \; \sum_x \big( f(I,\theta^x) - f(I,\theta) \big),
\end{equation}
with $\theta^x$ is obtained from $\theta$ by changing $\theta_x$ to
$-\theta_x$ ($-\theta_x$ is the inverse of $\theta_x$ for the
addition on $\T$).  
The symmetry of the potential $W$ implies 
\begin{equation*}
P(I_x,-\theta_x) \; = \; - P(I_x,\theta_x) \qquad \text{and} \qquad Q(I_x,-\theta_x) \; = \;  Q(I_x,\theta_x).
\end{equation*}
This implies that the noise $S$, as defined by \eqref{noise action angle}, preserves the total energy, and that the relation
\begin{equation*}
S j_{x,x+1} = - 4 j_{x,x+1}
\end{equation*}
holds. 

\subsection{The special case $W$ given by \eqref{special case W}.}
Let us now assume that $W(q)= |q|^r/r$, i.e. 
\begin{equation}\label{polynomial hamiltonian}
H_0 (q,p) \; = \; \frac{p^2}{2} + \frac{|q|^r}{r}, \qquad r > 2. 
\end{equation}
The following scaling relation are readily checked:
\begin{equation}\label{I and omega}
\mathrm H_0 (I) \; = \; \mathrm H_0 (1) \cdot  I^{2r/(r+2)}
\qquad \text{and} \qquad 
\omega (I) \; = \; \omega (1) \cdot I^{(r-2)/(r+2)}. 
\end{equation}
Moreover, writing 
\begin{equation}\label{Fourier Q P}
Q (I,\theta) \; = \; \sum_{k\in \Z} \hat Q (I,k) \ed^{\imaginary k \theta}, 
\qquad 
P (I,\theta) \; = \; \sum_{k\in \Z} \hat P (I,k) \ed^{\imaginary k \theta},
\end{equation}
we obtain 
\begin{equation}\label{Q P polynomial hamiltonian}
Q (I,\theta) \; = \; I^{2/ (r+2)} \sum_{k\in \Z} \hat Q (1,k) \ed^{\imaginary k \theta}, 
\qquad 
P (I,\theta) \; = \; I^{r/(r+2)} \sum_{k\in \Z} \hat P (1,k) \ed^{\imaginary k \theta}.
\end{equation}
Because $Q(1,\theta)$ and $P (1, \theta)$ are smooth, the Fourier coefficients $\hat Q (1,k)$, $\hat P (1,k)$, with $k \in \Z$, have good decay property as $|k| \rightarrow \infty$.

\subsection{Poisson equation for the uncoupled dynamics.}

In this subsection, we consider functions on $\R_+^2 \times \T^2$, that depend on two actions $(I_0,I_1)$ and two angles $(\theta_0,\theta_1)$. 
The actions play the role of a parameter, and, for clarity, will be dropped from several notations.
A function $f\in \mathcal C^\infty (\T^2)$ is expanded in Fourier series as
\begin{equation*}
f (I_0,I_1,\theta_0,\theta_1) 
\; = \;
\sum_{(k_0,k_1)\in \Z} \hat f (I_0,I_1,k_0,k_1) \ed^{\imaginary (k_0 \theta_0 + k_1 \theta_1)}
\end{equation*}
with
\begin{equation*}
\hat f (I_0,I_1,k_0,k_1) \; = \; \frac{1}{(2\pi)^2} \int_{[-\pi,\pi]^2} f (I_0,I_1,\theta_0,\theta_1) \ed^{-\imaginary (k_0 \theta_0 + k_1 \theta_1)} \, \dd \theta_0 \dd \theta_1.
\end{equation*}
It is seen that the current satisfies $\hat j_{0,1}(I_0,I_1,0,0) = 0$ for all $(I_0,I_1) \in \R_+^2$.
We introduce the notations
\begin{align*}
\eta (k_0,k_1) \; &= \;  \imaginary \big( k_0 \omega (I_0) + k_1 \omega (I_1) \big) - 2 \noise \\
D (k_0,k_1) \; &= \;   \eta(k_0,k_1) \eta (-k_0,-k_1) - \frac{16 \noise^4}{\eta (-k_0,k_1) \eta (k_0,-k_1)}.
\end{align*}

\begin{lem}\label{lem: Poisson action angle}
Let $f$ be a function on $\R_+^2 \times \T^2$ such that $f(I_0,I_1,\cdot,\cdot)$ is smooth and satisfies $\hat f (I_0,I_1,0,0) = 0$, for any $(I_0,I_1) \in \R_+^2$. 
Writing $\hat f(k_0,k_1)$ for $\hat f(I_0,I_1,k_0,k_1)$, we define
\begin{multline}\label{function g action angle}
g(I_0,I_1,k_0,k_1)
\; = \;
\hat f(k_0,k_1) 
\, - \, 
\noise \Big( \frac{\hat f(-k_0,k_1)}{\eta(-k_0,k_1)} + \frac{\hat f(k_0,-k_1)}{\eta(k_0,-k_1)} \Big) \\
+ \, 
\frac{\noise^2}{\eta(-k_0,-k_1)} \Big(\frac{1}{\eta(-k_0,k_1)} + \frac{1}{\eta(k_0,-k_1)} \Big) 
\big( \hat f(-k_0,-k_1) - \hat f(k_0,k_1) \big).
\end{multline}
A solution $u$ to the equation $-L_0 u = f$ is given, in the Fourier variables, by 
\begin{align*}
\hat u (I_0,I_1,0,0) 
\; &= \; 
0, \\
 \hat u (I_0,I_1,k_0,k_1) 
\; &= \; 
- \frac{\eta(-k_0,-k_1)}{D(k_0,k_1)}
\, g (I_0,I_1,k_0,k_1)
\quad \text{for} \quad 
(k_x,k_y) \ne (0,0).
\end{align*}
\end{lem}

\begin{proof}
In the Fourier variables, the equation $-L_0 u = f$ reads
\begin{equation*}
\eta (k_0,k_1) \hat u (k_0,k_1) + \noise \hat u (-k_0,k_1) + \noise \hat u (k_0,-k_1) \: = \; - \hat f(k_0,k_1)
\end{equation*}
where we have written $\hat u (k_0,k_1)$ for $\hat u (I_0,I_1,k_0,k_1)$. 
The result is then checked by means of a direct computation. 
\end{proof}

\noindent
\textbf{Remarks.}
1. All other solutions are obtained by taking for $\hat u (I_0,I_1,0,0)$ an arbitrary function of the actions $I_0,I_1$.
This choice is irrelevant for the sequel. 

\noindent
2. Since $| \noise / \eta (k_0,k_1) | \le 1$ for all $(k_0,k_1) \in \Z^2$, we have the bound
\begin{equation*}
|g(I_0,I_1,k_0,k_1)| \; \le \; 5 \max \{  |\hat f(k_0,\pm k_1)|, |\hat f(-k_0,\pm k_1)|\}.
\end{equation*}

\noindent
3. For $\noise = 0$, the solution simply becomes 
\begin{equation*}
\hat u (I_0, I_1,k_0,k_1) \; = \; \imaginary \frac{\hat f(k_0,k_1)}{k_0 \omega (I_0) + k_1 \omega (I_1)}
\qquad \text{for} \qquad 
(k_x,k_y) \ne (0,0).
\end{equation*}

\subsection{
Proof of Proposition \ref{prop:ubgeneric}
}

By (\ref{eq:k2ublb}) we have
\begin{equation*}
\beta^{-2}\kappa_2 (\noise) 
\; \le \;
\langle \gamma^2 (e_0,e_1) \rangle_{\beta,0}
\; = \; 
\langle j_{0,1} (-L_0)^{-1} j_{0,1} \rangle_{\beta,0},  
\end{equation*}
with $\langle \cdot \rangle_{\beta,0}$ the uncoupled Gibbs state. 
Writing $u = (-L_0)^{-1} j_{0,1}$ we have thus
\begin{align*}
\langle j_{0,1} (-L_0)^{-1} j_{0,1} \rangle_{\beta,0}
\; &\sim \;
\int_{\R_+^2}  \ed^{-\beta \big(  \mathrm H_0 (I_0) + \mathrm H_0 (I_1) \big)} \, \dd I_0 \dd I_1
\\
&\qquad \times \int_{\T^2} u(I_0,I_1,\theta_0,\theta_1) j_{0,1} (I_0,I_1,\theta_0,\theta_1) \, \dd \theta_0 \dd \theta_1\\
\; &\sim \; 
\int_{\R_+^2}  \ed^{-\beta \big(  \mathrm H_0 (I_0) + \mathrm H_0 (I_1) \big)} \, \dd I_0 \dd I_1\\
&\qquad \times \sum_{k_0,k_1 \in \Z^2} \hat u(I_0,I_1,k_0,k_1) \hat j_{0,1} (I_0,I_1,k_0,k_1).
\end{align*}
Writing 
\begin{equation}\label{function h action angle}
h(I_0,I_1,k_0,k_1) \; = \; - g(I_0,I_1,k_0,k_1) \hat j (I_0,I_1,k_0,k_1),
\end{equation}
with $g$ as defined in \eqref{function g action angle},
we obtain by Lemma \ref{lem: Poisson action angle}, 
\begin{equation}\label{integral for kappa 2}
  \begin{split}
   & \langle j_{0,1} (-L_0)^{-1} j_{0,1} \rangle_{\beta,0} \; \sim \;\\
   & \sum_{(k_0,k_1) \ne (0,0)} \int_{\R_+^2} \frac{\eta
      (-k_0,-k_1)}{D(k_0,k_1)} \, h (I_0,I_1,k_0,k_1) \, \ed^{-\beta
      \big( \mathrm H_0 (I_0) + \mathrm H_0 (I_1) \big)} \, \dd I_0
    \dd I_1.
  \end{split}
\end{equation}
In this expression,
\begin{equation}\label{expression alpha D}
\frac{\eta (-k_0,-k_1)}{D(k_0,k_1)} 
\; = \;
-
\frac{ 2 \noise \, + \, \imaginary \, \big(k_0 \omega (I_0) + k_1 \omega (I_1) \big)}
{\big(k_0 \omega (I_0) + k_1 \omega (I_1) \big)^2 
\, + \, 4\noise^2 \frac{\left(k_0 \omega (I_0) - k_1 \omega (I_1)\right)^2}{\left(k_0 \omega (I_0) - k_1 \omega (I_1)\right)^2 \, + \, 4 \noise^2}}.
\end{equation}

We now come to the crux of the argument, and start using the specific form of $H_0$. 
In view of \eqref{expression alpha D}, it looks desirable to change integration variables in \eqref{integral for kappa 2}
from $(I_0,I_1)$ to $(\omega_0,\omega_1) = (\omega(I_0),\omega(I_1))$. 
The anharmonicity of $W$, specifically expressed in this case by
relation \eqref{I and omega}, makes this possible, giving 
\begin{equation}\label{integral for kappa 2 bis}
  \begin{split}
    & \langle j_{0,1} (-L_0)^{-1} j_{0,1} \rangle_{\beta,0} \; \sim \\
    & \sum_{(k_0,k_1) \ne (0,0)} \int_{\R_+^2}
    \frac{\eta (-k_0,-k_1)}{D(k_0,k_1)} \, \tilde h
    (\omega_0,\omega_1,k_0,k_1) \, (\omega_0 \omega_1)^{\frac{4}{r-2}} \,
    \rho_\beta (\omega_0,\omega_1)\, \dd \omega_0 \dd \omega_1
  \end{split}
\end{equation}
with
\begin{align*}
\tilde h \big(\omega_0,\omega_1,k_0,k_1) \; &= \;  h \big( c(r)\omega_0^{\frac{r+2}{r-2}}, c(r) \omega_1^{\frac{r+2}{r-2}} , k_0,k_1 \big), \qquad c(r) > 0, \\
\rho_\beta (\omega_0,\omega_1) \; &= \; \ed^{- c'(r) \beta (  \omega_0^{\frac{2r}{r-2}} + \omega_1^{\frac{2r}{r-2}} )}, \qquad c'(r) > 0.
\end{align*}

To proceed, we need some more technical informations on the function $\tilde h (\omega_0,\omega_1,k_0,k_1)$. The potential $W$ is not strictly convex at the origin, implying that $\omega (I)$ vanishes as $I \rightarrow 0$. 
For this reason, we need a relatively detailed knowledge on $\tilde
h(\omega_0,\omega_1,k_0,k_1)$ for $(\omega_0,\omega_1)$ near the origin,  
in a order to exclude any divergence at small frequencies. 

Using the general expression \eqref{current action angle} for the
current $j_{0,1}$,  
the specific expression \eqref{Q P polynomial hamiltonian} for
$Q(I,\theta)$ and $P(I,\theta)$,  
the definition \eqref{function g action angle} of $g$, and the
definition \eqref{function h action angle} of $h$,  
we conclude that $h$ is of the form 
\begin{multline*}
h(I_0,I_1,k_0,k_1) 
\; = \; 
I_0^{\frac{2r}{r+2}}
h_{0,0}(I_0,I_1,k_0,k_1) 
\, + \, 
I_0^{\frac{r}{r+2}} I_1^{\frac{r}{r+2}}
h_{0,1}(I_0,I_1,k_0,k_1) 
\; \\
 + \; 
I_1^{\frac{2r}{r+2}} 
h_{1,1}(I_0,I_1,k_0,k_1), 
\end{multline*}
so that in turn $\tilde h$ takes the form 
\begin{multline}\label{refined expression tilde h}
\tilde h(\omega_0,\omega_1,k_0,k_1) 
\; = \; 
\omega_0^{2r/(r-2)}
\tilde h_{0,0}(\omega_0,\omega_1,k_0,k_1) \\
\, + \, 
\omega_0^{r/(r-2)} \omega_1^{r/(r-2)}
\tilde h_{0,1}(\omega_0,\omega_1,k_0,k_1) 
\; + \; 
\omega_1^{2r/(r-2)} 
\tilde h_{1,1}(\omega_0,\omega_1,k_0,k_1), 
\end{multline}
where $\tilde h_{i,j}$ satisfies the following bounds: 
there exists $a < + \infty$ and, for any $b > 0$, there exists a constant $\mathrm C_b < + \infty$, such that 
\begin{equation}\label{decay of h}
\tilde h_{i,j}(\omega_0,\omega_1,k_0,k_1)
\; \le \; 
\mathrm C_b \frac{(|\omega_0| + |\omega_1| + 1 )^a}{(|k_0| + |k_1| + 1)^b}, 
\quad (i,j) = (0,0), (0,1), (1,1).
\end{equation}
Moreover, by symmetry, we have $\hat P(I,0) = 0$ for all $I > 0$, with $\hat P(I,0)$ defined by \eqref{Fourier Q P}. 
It follows that 
\begin{equation}
\label{symmetry tilde h current}
\begin{split}
\tilde h_{0,0} (\omega_0,\omega_1, 0 ,k_1) 
\; &= \; 
\tilde h_{0,1} (\omega_0,\omega_1, 0 ,k_1) \\
\; &= \; 
\tilde h_{0,1} (\omega_0,\omega_1, k_0 ,0)
\; = \; 
\tilde h_{1,1} (\omega_0,\omega_1, k_0 ,0)
\; = \; 0.  
\end{split}
\end{equation}

We now move back to the evaluation of \eqref{integral for kappa 2 bis}. 
We distinguish three cases, according to the values of $k_0$ and $k_1$; resonances appear in case 3.
The sum over $(k_0,k_1) \in \Z^2/\{0,0\}$ can then be controlled thanks to the decay in \eqref{decay of h} with $b$ large enough.

\emph{Case 1:} $k_0 k_1= 0$.
Let us, as an example, consider the case $k_0 = 0, k_1 \ne 0$. 
The integral \eqref{integral for kappa 2 bis} has a possible
divergence only for $k_1 \rightarrow 0$.  
We have
\begin{equation*}
\Big|\frac{\eta (-k_0,-k_1)}{D(k_0,k_1)} \Big| \; \le \; \mathrm C \omega_1^{-2} \qquad \text{for} \qquad \omega_0,\omega_1 \le 1.
\end{equation*}
Thanks to \eqref{symmetry tilde h current}, only the term in $\tilde h_{1,1}$ survives in \eqref{refined expression tilde h}, and we conclude that the integrand behaves as 
\begin{equation*}
\omega_1^{4/(r-2)} \omega_1^{2r/(r-2)} \omega_1^{-2} = \omega_1^{8/(r-2)} \qquad \text{as} \qquad \omega_1 \rightarrow 0,
\end{equation*}
so that there is in fact no singularity. 

\emph{Case 2:} $k_0 k_1 > 0$. 
The only possible divergence of the integral \eqref{integral for kappa
  2 bis} is at the origin. 
We have the bounds
\begin{equation*}
\Big|\frac{\eta (-k_0,-k_1)}{D(k_0,k_1)} \Big| \; \le \; \mathrm C \omega_0^{-2}, \mathrm C \omega_0^{-1}\omega_1^{-1}, \mathrm C \omega_1^{-2} 
\qquad \text{for} \qquad \omega_0,\omega_1 \le 1, 
\end{equation*}
allowing to check, as in the previous case, that there is no singularity. 

\emph{Case 3:} $k_0 k_1 < 0$. The integrand now becomes truly singular (resonances).
Let us assume, for example, that $k_0 > 0$ and $k_1 < 0$. 
We split the integral \eqref{integral for kappa 2 bis} as 
\begin{equation}\label{split integral noise}
\int_{\R_+^2} (\dots) \; = \; \int_{k_0 \omega_0 + |k_1| \omega_1 < \noise} (\dots)  \, + \, \int_{k_0 \omega_0 + |k_1| \omega_1 \ge \noise} (\dots).
\end{equation}
For the first integral, we are satisfied by the rough bound
\begin{equation*}
\Big|\frac{\eta (-k_0,-k_1)}{D(k_0,k_1)} \Big| \; \le \;  \frac{\mathrm C}{\noise^2 (k_0 \omega_0 + |k_1|\omega_1)^2}.
\end{equation*}
As in the cases treated previously, it is seen that there is no singularity. 
Moreover, the integration domain is of size $\noise^2$, so that the integral is of order 1 at most. 

We move to the second integral. 
We find it convenient to change again variables. 
With
\begin{equation*}
x \; = \; k_0 \omega_0 + |k_1| \omega_1, 
\qquad 
y \; = \; k_0 \omega_0 - |k_1| \omega_1,
\end{equation*}
the second integral in the right hand side of \eqref{split integral noise} becomes
\begin{equation*}
\begin{split}
& \int_{k_0 \omega_0 + |k_1| \omega_1 \ge \noise} (\dots)\\
&\; \sim \; 
\int_{\noise}^\infty \dd x \int_{-x}^x \dd y \frac{2 \noise +   \imaginary y}{y^2 + 4\noise^2 \frac{x^2}{x^2 + 4 \noise^2}}\,
\phi (x,y,k_0,k_1)\, \tilde\rho_\beta (x,y,k_0,k_1)
\end{split}
\end{equation*}
with 
\begin{align*}
\phi (x,y,k_0,k_1) \; &= \; h\left( \frac{x+y}{2k_0}, \frac{x-y}{2|k_1|},k_0,k_1\right) \,  \left(\frac{(x+y)(x-y)}{4 k_0|k_1|}\right)^{4/(r-2)} \\
\tilde\rho_\beta (x,y,k_0,k_1) \; &= \; \rho_\beta \left(\frac{x+y}{2k_0},
\frac{x-y}{2|k_1|}\right). 
\end{align*}
We observe that, in the domain of integration $x\ge \noise$:
\begin{equation*}
4\noise^2 \frac{x^2}{x^2 + 4 \noise^2} \; \ge \; \frac 45 \noise^2.
\end{equation*}
Therefore, the integral converges to a finite value as $\noise \rightarrow 0$. $\square$

\section{Proof of Corollary \ref{cor: dhc corol}}\label{Appendix proof dhc}
Consider the quenched space-time correlations of the energy:
\begin{equation*}
  S(x,t, \omega) \ = \left<e_x(t) e_0(0) \right>_{\rho_\beta} - \beta^{-2}
\end{equation*}
where $\omega = (\omega_x)_{x\in \Z}$, and where $\{e_x(t)\}$ is the time evolved energy generated by the
Ginzburg-Landau dynamics $\mc L$ 
with the coefficients $\gamma^2$ and $\alpha$ computed above, starting
with the equilibrium distribution at temperature $\beta^{-1}$.
Computing the time derivative, we have
\begin{equation*}
  \begin{split}
    \partial_t S(x,t, \omega) \ =\  8\noise \Delta_{x+1,x, \omega}^{-1} 
      \left[S(x+1,t,\omega) - S(x,t,\omega)\right]  \\
      -8\noise \Delta_{x,x-1, \omega}^{-1}
       \left[S(x,t,\omega) - S(x-1,t,\omega)\right] 
  \end{split}
\end{equation*}
with $\Delta_{x,y,\omega}$ defined by \eqref{Delta Definition}.
Thus $S(x,t,\omega) = \mathbb E_{0,\omega}(\delta_x(X(t)))$, the
transition probability of a
1-dimensional random walk on random bonds  $X(t)$ (so called bond
diffusion).  
It is well known and easy to compute the asymptotic variance of this
bond diffusion, it is given by the harmonic average of the bonds
variables (\cite{KLO}):
\begin{equation}
  \label{eq:19}
  \begin{split}
    \lim_{t\to\infty} \frac 1t \sum_x x^2 S(x,t,\omega)& =
    \lim_{t\to\infty} \frac 1t \mathbb E_{0,\omega}(X(t)^2)\\& =
    \Big\langle \Big(
    \frac{8\noise}{\Delta_{0,1}(\varsigma)}\Big)^{-1}
    \Big\rangle^{-1}_* \;= \; \frac{8\noise}{\langle
      \Delta_{0,1}(\varsigma)\rangle_*}
  \end{split}
\end{equation}
almost surely in $\omega$. 

By the Green-Kubo formula for the diffusivity for $\mc L$, this is yields
\begin{equation}
\kappa_2 (\noise) \; = \; \frac{8\noise\beta^2}{\langle
  \Delta_{0,1}(\varsigma)\rangle_*}  
\; \rightarrow \; 0 \qquad \text{as} \qquad \noise \rightarrow 0, 
\end{equation}
which gives the claims. $\square$

\section{Proof of Proposition \ref{prop: weak coupling rotors}}\label{Appendix: proof of a lemma and a proposition}

We start by the following lemma.
\begin{lem}\label{lem: inverse of L0 rotators}
Let $x,y\in \bZ$. 
A solution $\psi_{x,y}$ to the equation
\begin{equation}
- L_0 \psi_{x,y} = \sin(q_x- q_{y}) p_x
\end{equation}
is given by
\begin{equation}
\begin{split}
\psi_{x,y} \; =\; & \Delta_{x,y}^{-1} \left\{[ 4 \noise^2 +(e_{x} -e_{y})] e_x + \frac12 (e_x -e_{y}) p_x p_{y} \right\} \cos (q_x -q_y) \\
&+ \Delta_{x,y}^{-1} \left\{ 2 \noise (e_{y} p_x + e_x p_{y}) \right\} \sin(q_{x} -q_{y})
\end{split} 
\end{equation}
with 
\begin{equation}
\Delta_{x,y} \; := \; \Delta (e_x,e_{y}) \; = \;  4 \noise^2 (e_x +e_{y}) + (e_{y} -e_x)^2.
\end{equation}
\end{lem} 

\begin{proof}
We compute
\begin{align*}
& A \psi_{x,y}
\; =\;  2 \noise \Delta_{x,y}^{-1}
 (e_y p_x + e_x p_y) (p_x - p_y) \cos (q_x - q_y) \\
& \quad -\Delta_{x,y}^{-1}
\Big\{
\big(4\noise^2 + (e_x - e_y)\big) e_x + \frac{1}{2} (e_x - e_y) p_x p_y
\Big\}
(p_x - p_y) \sin (q_x - q_y) \\
\end{align*}
and
\begin{equation*}
S \psi_{x,y}
\; = \;
- 4\noise \Delta_{x,y}^{-1}   (e_y p_x + e_x p_y) \sin (q_x - q_y)
 - 2\Delta_{x,y}^{-1}  (e_x - e_y) p_x p_y \cos (q_x - q_y).
\end{equation*}
Remembering that $p_x^2 = 2e_x$ and $p_y^2=2e_y$, the terms in $\cos (q_x - q_y)$ cancel in $(A + \noise S)\psi_{x,y}$, so that 
\begin{equation*}
[A + \noise S ] \psi_{x,y} \; = \; \Delta^{-1}_{x,y} \theta_{x,y}\sin(q_{x} -q_y)
\end{equation*}
with
\begin{align*}
\theta_{x,y}
\; &= \; 
\Big\{
\big(4\noise^2 + (e_x - e_y)\big) e_x + \frac{1}{2} (e_x - e_y) p_xp_y
\Big\} 
(p_x - p_y)\\
&\quad - 4 \noise^2  (e_y p_x + e_x p_y)\\
& = \; - p_x \Delta_{x,y}.
\end{align*}
This proves the claim.
\end{proof}

\begin{proof}[{\bf Proof of Proposition \ref{prop: weak coupling rotors}}]
The Gibbs measure at inverse temperature $\beta$ is readily computed. 
For a function $f$ depending only on the uncoupled energy $e_x = p_x^2/2$, it holds that 
\begin{equation*}
\langle f \rangle_{\beta,0}
\; = \; 
\sqrt{\frac{\beta}{2\pi}} \int_\bR f(p_x^2/2) \ed^{-\beta p_x^2/2} \, \dd p_x
\; = \;
\sqrt{\frac{\beta}{4\pi}} \int_0^\infty f (e) \ed^{-\beta e} \,\frac{\dd e}{\sqrt e}
\end{equation*}
from which \eqref{Gibbs measure rotors} follows. 

Next, $\gamma^2 (e_x,e_{x+1})$ is computed by means of Lemma  \ref{lem: inverse of L0 rotators}: 
\begin{align*}
& \gamma^2 (e_x,e_{x+1})
\; = \; \Pi \,  j_{x,x+1} (-L_0)^{-1} j_{x,x+1} \\
\;  &\quad = \; \frac{1}{2}\Pi \,  j_{x,x+1} \big[ (-L_0)^{-1} \sin (q_{x+1} - q_{x}) p_{x+1} - (-L_0)^{-1} \sin(q_{x} - q_{x+1}) p_{x} \big] \\
\; &\quad = \; \frac{1}{2} \Pi \, j_{x,x+1} \big[ \psi_{x+1,x} - \psi_{x,x+1} \big].
\end{align*}
The terms in $\cos (q_x - q_{x+1})$ in $\psi_{x+1,x}$ and $\psi_{x,x+1}$ will vanish due to the projection $\Pi$, so that we are left with 
\begin{equation*}
\gamma^2 (e_x,e_{x+1}) =
\frac{1}{4} \Pi (p_x + p_{x+1}) \sin (q_{x+1} - q_x)
\frac{4 \noise (e_{x+1}p_x + e_x p_{x+1}) \sin (q_{x+1} - q_{x})}{ \Delta_{x,x+1} } .
\end{equation*}
Since $\frac{1}{(2\pi)^2}\int_{[0,2\pi]^2} \sin^2 (x - y) \, \dd x \dd y = 1/2$,
and since the projection of expressions containing uneven powers of $p_x$ or $p_{x+1}$ vanishes, we obtain \eqref{gamma rotors}.

The current $\alpha (e_x,e_{x+1})$ can be computed in two possible ways: 
directly by the definition $\alpha (e_x,x_{x+1}) = \Pi G (-L_0)^{-1} j_{x,x+1}$, 
or by means of the expression
\begin{equation*}
\alpha (e_x,e_{x+1})
\; = \; 
\ed^{\mathcal U (e_x) + \mathcal U (e_{x+1})} \big(\partial_{e_{x+1}} - \partial_{e_x}\big) \ed^{ - (\mathcal U (e_x) + \mathcal U (e_{x+1}))} \gamma^2 (e_x,e_{x+1})
\end{equation*} 
with $\mathcal U(x) = \frac{1}{2}\log x$. 
Both computations lead to \eqref{alpha rotors}. 
\end{proof}



\begin{thebibliography}{999}
\footnotesize

\bibitem{B08} C. Bernardin, Stationary non-equilibrium properties for a heat conduction model, Phys. Rev E, \textbf{78}, no. 2, (2008).

\bibitem{CBFH} C. Bernardin and F. Huveneers, 
Small perturbation of a disordered harmonic chain by a noise and an anharmonic potential, 
Prob. Th. Rel. Fields {\bf 157}, 301-331 (2013).

\bibitem{BO11} C. Bernardin and S. Olla, 
Transport Properties of a Chain of Anharmonic Oscillators with Random
Flip of Velocities, 
J. Stat.Phys. {\bf 145}, 1224-1255 (2012). 

\bibitem{BObook}  C. Bernardin and S. Olla, 
Thermodynamics and Non-equilibrium Macroscopic Dynamics of Chains
Anharmonic Oscillators, 
https://www.ceremade.dauphine.fr/~olla/springs13.pdf


\bibitem{blr} F. Bonetto, J.L. Lebowitz, L. Rey-bellet, Fourier's Law:
  a challenge to theorist, Mathematical Physics 2000, Imperial College Press, London, 2000, pp.128-150.

\bibitem{BL07}
Oliver Butterley and Carlangelo Liverani,
 Smooth anosov flows: correlation spectra and stability.
 {Journal of Modern Dynamics}, {\bf 1}, 2:301-322, 2007. 
 
\bibitem{BL13} Oliver Butterley and Carlangelo Liverani, Robustly
  invariant sets in fiber contracting bundle flows. Journal of Modern
  Dynamics, {\bf 7}, 2, 255--267 (2013) 

\bibitem{Car} P. Carmona, Existence and uniqueness of an invariant
  measure for a chain of oscillators in contact with two heat baths,
  Stochastic Process. Appl. {\bf 117}, no. 8, 1076--1092 (2007). 

\bibitem{dhar} A. Dhar, Heat Transport in low-dimensional systems,
 Advances in Physics Vol. 57, No. 5, 2008, 457--537.

\bibitem{DLK} A. Dhar, J.L. Lebowitz, V. Kannan, Heat conduction in
  disordered harmonic lattices with energy conserving noise,
  Phys. rev. E , 83, 021108, (2011), 

\bibitem{WDFH} W. De Roeck and F. Huveneers, 
Asymptotic localization of energy in non-disordered oscillator chains,
arXiv:1305.5127, 1-31 (2013).

\bibitem{DL} D. Dolgopyat and C. Liverani, 
Energy transfer in a fast-slow Hamiltonian system,
Commun. Math. Phys. {\bf 308}, 201--225 (2011).

\bibitem{EH} J.-P. Eckmann and M. Hairer, 
Non-Equilibrium Statistical Mechanics of Strongly Anharmonic Chains of Oscillators,  
Commun. Math. Phys. {\bf 212}, 105-164 (2000).



\bibitem{F2} J. Fritz, Gradient dynamics of infinite point systems,
Ann. Prob. {\bf 15} (1987), 478--514.  


\bibitem{FFL} {{J. Fritz, T. Funaki and J.L. Lebowitz}}, {Stationary
    states of random Hamiltonian systems},  Probab. Theory Related
  Fields {\bf 99}, (1994), 211--236.  

\bibitem{GAGI}
P. Gaspard and T. Gilbert,
Heat conduction and Fourier's law by consecutive local mixing and thermalization   
Physical Review Letters 101, 020601 (2008).

\bibitem{GL06}
S\'ebastien Gou\" ezel and Carlangelo Liverani,
 Banach spaces adapted to Anosov systems.
 {Ergodic Theory and Dynamical Systems}, {\bf 26}, 1:189--217, 2006.

\bibitem{Huveneers}
F.~Huveneers,
Energy fluctuations in simple conduction models, 
Stochastic Processes and their Applications, {\bf 123} (10), 3753--3769, 2013.

\bibitem{KV} C. Kipnis and S.R.S. Varadhan, Central limit theorem for
  additive functionals of reversible Markov processes and applications
  to simple exclusions, Comm. Math. Phys. {\bf{104}}, Number 1, 1-19
  (1986). 

\bibitem{KLO} T. Komorowski, C. Landim, and S. Olla,
  Fluctuations in Markov processes. Time symmetry and martingale
  approximation. Springer (2012). 



\bibitem{LLL} O.E. Lanford III, J.L. Lebowitz, E. H. Lieb, Time evolution of infinite anharmonic systems. 
J. Statist. Phys. {\bf 16}, no. 6, 453--461 (1977). 

\bibitem{LLP} S. Lepri, R. Livi, A. Politi, 
Thermal Conduction in classical low-dimensional lattices, 
Phys. Rep. {\bf 377}, 1--80 (2003).

\bibitem{Liv} C. Liverani,  On contact Anosov flows. Ann. of Math. (2) {\bf 159}, no. 3, 1275--1312 (2004).
\bibitem{LO} C. Liverani and S. Olla,  Toward the Fourier law for a weakly interacting anharmonic crystal,  JAMS {\bf 25}, N. 2, 555--583, (2012).

\bibitem{LOS} C. Liverani, S. Olla, M. Sasada, 
Macroscopic Energy diffusion after weak coupling, in preparation.

\bibitem{MPP}
C. Marchioro, A. Pellegrinotti, E. Presutti, 
Existence of time evolution in $\nu$-dimensional statistical mechanics, 
Commun. Math. Phys. {\bf 40}, 175-185 (1975).  

\bibitem{OT} S. Olla, C. Tremoulet, Equilibrium fluctuations for
  interacting Ornstein-Uhlenbeck particles, Comm. Math. Phys.
   {\bf 233} (3), 463-491 (2003). 

 \bibitem{Otto} Georgii, Hans-Otto, Gibbs measures and phase transitions. Second edition. de Gruyter Studies in Mathematics, 9. Walter de Gruyter \& Co., Berlin, 2011. xiv+545 pp.

\bibitem{Luc} L. Rey-Bellet,  Open Classical System, LNM XX, Springer.

\bibitem{S0} H. Spohn, Equilibrium Fluctuations for Interacting
  Brownian Particles, Comm. Math. Phys, {\bf 103}, 1-33, (1986).
 
\bibitem{S} H. Spohn, private communication.



\bibitem{Var} S.R.S. Varadhan, 
Nonlinear diffusion limit for a system with nearest neighbor
interactions II, Asymptotic problems in probability theory: stochastic
models and diffusions on fractals (Sanda/Kyoto, 1990), 75-128,  Pitman
Res. Notes Math. Ser., {\bf 283}, Longman Sci. Tech., Harlow, (1993).   

\bibitem{Zw} M.Zworski, private communication.

\end{thebibliography}
\end{document}